\documentclass[acmsmall,screen,nonacm]{acmart}

\usepackage[ruled]{algorithm2e}
\usepackage{orcidlink} 
\usepackage{multirow}
\theoremstyle{acmplain}
\newtheorem{theorem}{Theorem}[section]
\newtheorem{corollary}{Corollary}[theorem]
\newtheorem{lemma}[theorem]{Lemma}
\newtheorem*{remark}{Remark}
\usepackage{afterpage,lipsum,hyperref,graphicx}
\AtBeginDocument{%
  \providecommand\BibTeX{{%
    \normalfont B\kern-0.5em{\scshape i\kern-0.25em b}\kern-0.8em\TeX}}}

\setcopyright{none}
\copyrightyear{2018}
\acmYear{2018}
\acmDOI{XXXXXXX.XXXXXXX}

\begin{document}

\title{SP$_{RE}$V}


\author{Srivathsan Amruth\orcidlink{https://orcid.org/0009-0001-7568-7279}}
\affiliation{%
 \institution{National University Of Singapore}
 \country{amruth@u.nus.edu}}

\renewcommand{\shortauthors}{}

\begin{abstract}
SPREV, denoting \textit{(hyper)\textbf{S}phere \textbf{RE}duced to two-dimensional \textbf{RE}gular \textbf{P}olygon for \textbf{V}isualisation}, is a novel dimensionality reduction technique developed to addresses the challenges presented by reducing the dimension and data visualisation of labelled datasets characterized by the convergence of trifecta of characteristics—small class size, high dimensionality, and low sample size. SPREV aims not only to reveal but also to visually represent hidden patterns latent within such datasets. The methodology's unique blend of geometric principles tuned for use in discrete computational settings positions it as a indispensable tool for the modern data science toolkit, empowering users to discern trends, glean insights, and navigate the intricacies of their data promptly and efficiently.
\end{abstract}


\maketitle

\section{Introduction}

In contemporary research landscapes, the advent of developing sophisticated technologies driven by a heightened emphasis on precise contextual analytics but impeded by elevated data procurement costs stemming from legal and ethical considerations\cite{crow2006research}\cite{krotov2023big}\cite{richards2014big} surrounding data collection and storage\cite{klose2020edm} has ushered in an era\cite{sen2013data} where datasets exhibit intricate structures characterized by small class sizes, high dimensions, and limited sample sizes. This unique convergence of factors poses substantial challenges\cite{safo2014design} to the effective analysis, interpretation, and visualization of data, particularly in fields ranging from biology\cite{mahmud2019high}\cite{shen2023data}\cite{shen2024ultra} to finance\cite{choi2021boosting}\cite{6974480}\cite{saha2009small}. As researchers grapple with these complex datasets, the need for innovative methodologies becomes increasingly evident.

The quintessential problem of navigating through datasets characterized by a small number of classes, a high feature space, and a paucity of samples has garnered attention across diverse domains\cite{jackson2001adaptive}\cite{sun2022effectiveness}\cite{wang2022enhanced}. The challenges inherent in such datasets extend beyond conventional methodologies, often leading to issues of overfitting (often termed as data-piling\cite{shen2022classification} in the literature regarding this topic), increased computational complexity\cite{marron2007distance}, and limited interpretability of results\cite{ahn2007high}\cite{shen2016statistics}\cite{shen2022classification}. Acknowledging the exigency of addressing these challenges, our paper seeks to illuminate the underexplored realm of small class size, high dimension, and low sample size problems.


Examining the application of microarrays for gene expression measurement\cite{eisen199912}\cite{harrington2000monitoring}, a single measurement simultaneously reveals expression levels for $10^3$ to $10^4$ genes. Due to the substantial cost associated with these measurements, most datasets are constrained to modest sizes, typically in the range of $10^1$ to $10^2$ observations. Consequently, the dimensionality of the data, far exceeds the sample size. 

Considering another instance within the domain of biology, in medical diagnostics, the study of rare diseases often encounters a formidable barrier to assembling extensive datasets\cite{konietschke2021small}. The inherent scarcity of patients afflicted with these conditions limits the availability of labeled data, complicating the development of robust models for diagnostic and prognostic purposes. The primary obstacle, often, is rooted in the prohibitive costs associated with the procurement and annotation of an adequate dataset, exacerbated by the integration of high-dimensional imaging modalities that are required to furnish nuanced insights. This paradigm shift underscores the importance of tailored methodologies to navigate the complexities arising from restricted sample sizes and high dimensionality, offering a nuanced understanding of the opportunities and obstacles.

Similarly, in customer churn, where the focus lies on identifying patterns and trends within a confined set of classes, challenges\cite{jamjoom2021use} arise when these classes are inherently limited alongside limited samples. Looking at private banking\cite{ali2014dynamic} for this, understanding consumer behavior across a small set of target demographics due to the small niche client lists in tiny market segments demands intricate analysis of sparse datasets, where the challenge is not only the limited number of classes but also the necessity to discern subtle relationships among them. This complexity arises from the abundance of information collected per customer during the enrollment process for services offered by such service providers, resulting in a dataset characterized by high dimensionality per sample. Consequently, this confluence of factors necessitates a nuanced examination of methodologies tailored to the distinctive demands imposed by limited sample sizes, high dimensionality, and the exploration of relationships among a few pivotal classes. This approach fosters a deeper understanding of the challenges and opportunities inherent in such research focus across diverse domains.

In the face of escalating demands for data-driven insights in the modern era, SPREV emerges as a tool poised to provide a fresh perspective. Its tailored approach to dimension reduction and visualization in the context of low sample size and small class size problems positions it as a valuable addition to the modern data science toolkit. By addressing the unique challenges posed by contemporary datasets, SPREV has the potential to advance our analytical capabilities, opening avenues for exploration and discovery in the intricate landscapes of small class size, high-dimensional, low sample size data.

\section{Complications surfaced by Low Sample Size}

\subsection{Statistical Power and Reliability:}
One of the primary challenges associated with low sample sizes is the diminished statistical power, limiting the ability to detect true effects or relationships in the data\cite{button2013power}. In research and data analysis, statistical power is essential for making reliable inferences. With small sample sizes, the likelihood of obtaining statistically significant results is reduced, leading to less reliable and conclusive findings. For readers who are not acquainted with the terminology, statistical power\cite{cohen1992statistical} is the ability of the test to detect a true effect if it exists, indicating the likelihood that a significance test will correctly reject a false null hypothesis. A higher statistical power is desirable as it reduces the likelihood of Type II errors, where a true effect is overlooked. Type II errors\cite{banerjee2009hypothesis} are errors that occur when a statistical test fails to reject a false null hypothesis, thereby overlooking a genuine effect.

\subsection{Inadequate Precision and Confidence Intervals:}
The precision of estimates and confidence intervals is compromised with low sample sizes\cite{landis1977measurement}. Larger sample sizes provide more precise estimates of population parameters and narrower confidence intervals. In contrast, small sample sizes result in less accurate estimations and wider confidence intervals, limiting the precision of statistical inferences and making it challenging to ascertain the true nature of relationships within the data.

\subsection{Underpowered Hypothesis Testing:}
As mentioned earlier, low sample sizes reduce the power of statistical tests, and as a corollary increasing the risk of Type II errors\cite{cohen2013statistical}—failing to detect true effects. In hypothesis testing, where the goal is to make informed decisions about population parameters based on sample data, underpowered analyses undermine the reliability of conclusions. 

\subsection{Difficulty Detecting Small Effects:}
Small sample sizes make it challenging to detect small or subtle effects within the data\cite{raudys1991small}. In many disciplines, including data analysis, the identification of nuanced patterns or trends is crucial for drawing meaningful conclusions. In the absence of sufficient sample size, researchers may overlook important relationships, impacting the depth and accuracy of insights derived from the analysis.

\subsection{Generalisability and External Validity:}
Low sample sizes often struggle to describe the whole population and hence, result in non-representative sampling, compromising the external validity of a study\cite{asadoorian2005essentials}. In the context of data analysis, the ability to generalize findings to a larger population is hindered when the sample size is insufficient. This limitation restricts the broader applicability of research outcomes and may lead to inappropriate generalizations.

\subsection{Implications for Experimental and Observational Studies:}
 Both experimental and observational studies are affected by low sample sizes\cite{hertzog2008considerations}\cite{easterbrook1991publication}. In experimental research, small sample sizes may compromise the internal validity of experiments, limiting the ability to establish causal relationships. In observational studies, the representativeness of the sample is crucial for making unbiased inferences about the wider population, and small sample sizes impede this goal. Researchers may fail to identify significant effects or relationships that exist in the population due to inadequate sample sizes.

In summary, low sample size and small class size pose considerable impediments across diverse facets of research and data analysis. Researchers must exercise caution regarding these constraints to uphold the veracity and dependability of their findings. They ought to contemplate strategies for mitigating these challenges, all the while recognizing the plausible influence of limited class sizes on data quality. In alignment with this perspective, the assumptions articulated in our methodology are attuned to these considerations, crafted to accommodate for such limitations.

\section{Why Do We Care About High Dimensions}

\subsection{Introduction to High Dimensions}
In contemporary data-driven landscapes, the proliferation of high-dimensional datasets has become emblematic of our capacity to capture nuanced and intricate information across various domains. High-dimensional data refers to datasets where the number of features or variables vastly exceeds the number of observations, presenting an inherent richness that mirrors the complexity of real-world systems. This proliferation of dimensions occurs naturally in diverse scientific, technological, and societal contexts, underscoring the ubiquity of high-dimensional challenges.

Real-world phenomena often exhibit multifaceted characteristics, and as data collection technologies advance, the ability to capture intricate details increases. For instance, in genomics, high-throughput sequencing technologies yield datasets with dimensions corresponding to the multitude of genes, regulatory elements, and molecular interactions, providing a comprehensive snapshot of cellular processes. Similarly, in image and signal processing, the pixel intensity values across different channels or sensors contribute to high-dimensional representations of visual or sensory information. Moreover, in economic and financial analyses, the inclusion of numerous economic indicators and market variables results in datasets with dimensions reflecting the multifaceted nature of economic systems.

The emergence of high-dimensional datasets reflects the increasing granularity and complexity with which we probe the intricacies of our world. As dimensions accrue, however, the analytical landscape undergoes a fundamental transformation, necessitating a departure from conventional methodologies to address the unique challenges that manifest in these expansive spaces.

\subsection{The Curse of Dimensionality}

At the heart of the challenges posed by high-dimensional data lies the curse of dimensionality, a phenomenon marked by the exponential increase in the volume of the data space as the number of features grows. This results in sparsity, as data points become increasingly scattered, diluting the representation of the underlying structure. The curse of dimensionality also manifests in the diminishing efficacy of traditional statistical methods, rendering them inadequate for capturing meaningful patterns and relationships. Consequently, as dimensions proliferate, the reliability and interpretability of analyses conducted within these expansive spaces become compromised.

\subsection{Existence of Near-Orthogonality in Higher Dimensions}
\label{section:Near-Orthogonality}
Random vectors tend to become approximately orthogonal as the dimensionality increases. This phenomenon is often referred to as the "concentration of measure" or the "curse of dimensionality". To understand why high-dimensional random vectors become almost orthogonal, consider the following intuition: Orthogonality implies independence between vectors, and is measured by the inner product being zero. In high-dimensional spaces, the likelihood of two random vectors being exactly orthogonal decreases, but they become increasingly close to orthogonal (inner product approaches zero), since random components in each dimension have a diminished influence on the overall inner product, due to the vast number of possible configurations high-dimensional vectors can have. More informally, having more dimensions means more ways for random vectors to end up nearly independent, and hence nearly orthogonal.

Therefore, high-dimensional data produces near-orthogonality, which introduces the challenges associated with decreasing explainability in the realm of medical data analysis. We will now prove the existence of this "curse of dimensionality", starting with introducing a few preliminary results.

\subsubsection{Preliminary results needed for showing the existence of near-orthogonality.} 
\label{subsubsection: Prelim Results for northo}
These are the notations we will use throughout this paper. Let $\vec{u},\vec{v} \in \mathbb{R}^{n}$ be two random vectors, and let $\vec{u'},\vec{v'}$ be their unit vectors (the norm is 1), so $\vec{u} = \left |\vec{u}  \right | \vec{u'}$ and $\vec{v} = \left |\vec{v}  \right | \vec{v'}$. Furthermore, $\textup{P}\left ( \cdot  \right )$ denotes probability and $\textup{E}\left ( \cdot  \right )$ denotes expectation.

\begin{lemma}
\label{lemma:n-ortho one}
$\forall n\in \mathbb{N}\::  e^{\sqrt{n}}> n$.
\end{lemma}
This result is trivial to verify, e.g. checking that the graph of $y = e^{\sqrt{n}}-n$ lies above the positive x-axis.

\begin{proposition}
\label{proposition:random-vec-orthogonal-unit-vect}
Two random vectors $\vec{u},\vec{v}$ are approximately orthogonal if and only if their unit vectors are also approximately orthogonal.
\end{proposition}

\begin{remark}
Orthogonality between two vectors $\vec{u},\vec{v}$ is equivalent to their inner product being zero, i.e. $\left \langle \vec{u} , \vec{v} \right \rangle = 0$. We can verify approximate orthogonality (near-orthogonality) between two vectors $\vec{u},\vec{v}$ similarly by checking that $\left \langle \vec{u} , \vec{v} \right \rangle \approx  0$, or equivalently that $\left \langle \vec{u} , \vec{v} \right \rangle \to  0$.
\end{remark}

\begin{proof}
Assume that $\left \langle \vec{u} , \vec{v} \right \rangle \approx  0$. Here, $\left \langle \vec{u} , \vec{v} \right \rangle = \left \langle \left |\vec{u}  \right | \vec{u'} , \left |\vec{v}  \right | \vec{v'} \right \rangle = \left | \vec{u} \right |\left | \vec{v} \right | \left \langle \vec{u'} , \vec{v'} \right \rangle$, so $\left \langle \vec{u'} , \vec{v'} \right \rangle = \frac{\left \langle \vec{u} , \vec{v} \right \rangle}{\left | \vec{u} \right |\left | \vec{v} \right |} \approx  \frac{0}{\left | \vec{u} \right |\left | \vec{v} \right |} = 0$.
If instead we assume that $\left \langle \vec{u'} , \vec{v'} \right \rangle \approx 0$, then similarly $\left \langle \vec{u} , \vec{v} \right \rangle = \left | \vec{u} \right |\left | \vec{v} \right | \left \langle \vec{u'} , \vec{v'} \right \rangle \approx 0$. Thus both statements of Proposition \ref{proposition:random-vec-orthogonal-unit-vect} are biconditional as desired.
\end{proof}

\begin{theorem}
\label{Hoeffdings Inequality}
Hoeffding's inequality \cite{doi:10.1080/01621459.1963.10500830}: Given $n$ independent random variables $X_{1}, \cdots , X_{n}$ with $X_i \in \left [ a_i,b_i \right ]$ at least almost surely, the sum of random variables $S_{n} := \sum_{i=1}^{n}X_{n}$ satisfies the following inequality:
\begin{equation}
\label{equation 1 of northo}
\forall \varepsilon \in \mathbb{R}^{+}:\textup{P}\left [ \left | S_{n} - \textup{E}\left ( S_{n} \right ) \right | \geq   \varepsilon \right ] \leq 2\, \textup{exp}\left [ -\frac{2\varepsilon^{2}}{\sum_{k=1}^{n}\left ( b_{k}-a_{k} \right )^{2}} \right ]
\end{equation}
\end{theorem}

\subsubsection{Proof of near-orthogonality in higher dimensions.}We formally claim in the upcoming Proposition \ref{proposition: as n grows, inner product tends to 0} that for large $n \in \mathbb{N}$, random vectors in $\mathbb{R}^n$ tend to become approximately orthogonal. The mathematically-equivalent statement is that as dimensionality increases, there must exist a more-than-proportional number of unit vectors (and thus, vectors) that can exhibit near-orthogonality. Our proof of Proposition \ref{proposition: as n grows, inner product tends to 0} begins from the preliminary results in Section \ref{subsubsection: Prelim Results for northo}, and after the proof, we demonstrate this "curse of dimensionality" with a simulation in Section \ref{subsubsection: Simulation for norhto in high dim}.

\begin{proposition}
\label{proposition: as n grows, inner product tends to 0}
\textup{As} $n \to \infty $, $\exists \:  \vec{v}^1,\cdots, \vec{v}^m \in \mathbb{R}^n$ \textup{(where} $m > n$\textup{) such that} $\forall i,j \in \mathbb{N}_{\leq m}: \left \langle \vec{v}^j,\vec{v}^i \right \rangle \to 0$.
\end{proposition}

\begin{proof}
Let $\varepsilon = \sqrt{\frac{5}{\sqrt{n}}}$ and $m = e^{\sqrt{n}}$. From Lemma \ref{lemma:n-ortho one} we guarantee that $m > n$. Note that the values of $\varepsilon$ and $m$ depend on $n$, and $n$ itself approaches infinity in Proposition \ref{proposition: as n grows, inner product tends to 0}. Next, by Proposition \ref{proposition:random-vec-orthogonal-unit-vect}, it suffices to consider unit vectors when checking for near-orthogonality, so we define a set of $m$ randomly-generated unit vectors $\vec{v}^1, \cdots, \vec{v}^m$ as follows:
\newline $\forall i \in \mathbb{N}_{\leq m}$, let each component of $\vec{v}^i$ be $v_{k}^{i}$ ($k \in \mathbb{N}_{\leq n}$) with each $v_{k}^{i}$ chosen independently and randomly from the set $\left \{ -\frac{1}{\sqrt{n}},\frac{1}{\sqrt{n}} \right \}$, so that $\textup{P}\left ( v_{k}^{i} = -\frac{1}{\sqrt{n}}\right )=\textup{P}\left ( v_{k}^{i} = \frac{1}{\sqrt{n}}\right )=\frac{1}{2}$. Then, for each component of any of these random unit vectors, we have

\begin{equation}
\label{equation 2 of northo}
\textup{E}\left ( v_{k}^{i} \right ) = \frac{1}{2}\left ( -\frac{1}{\sqrt{n}} \right )+\frac{1}{2}\left (\frac{1}{\sqrt{n}} \right ) = 0
\end{equation}

Now, consider an arbitrary unit vector $u\in \mathbb{R}^n$, with $u_k$ representing its k-th component. We know from (\ref{equation 2 of northo}) that $\textup{E}\left ( v_{k}^{i} \right ) = 0$, so

\begin{equation}
\label{equation 3 of northo}
\textup{E}\left ( \left \langle \vec{u},\vec{v}^i \right \rangle \right ) = \textup{E}\left [ \sum_{k=1}^{n} 
\left ( u_k\cdot v_{k}^{i} \right ) \right ] = \textup{E}\left ( u_k \sum_{k=1}^{n} v_{k}^{i} \right ) = u_k \textup{E}\left ( \sum_{k=1}^{n} v_{k}^{i} \right ) = u_k \sum_{k=1}^{n} \textup{E}\left ( v_{k}^{i} \right ) = u_k \sum_{k=1}^{n} 0 = 0
\end{equation}

Moreover, we know $v_{k}^{i}$ can only take the values $-\frac{1}{\sqrt{n}}$ or $\frac{1}{\sqrt{n}}$, so it is clear that $u_k \cdot v_{k}^{i} \in \left [ -\frac{u_k}{\sqrt{n}},\frac{u_k}{\sqrt{n}} \right ]$. These conditions allow us to apply Hoeffding's inequality by setting (with reference to Theorem \ref{Hoeffdings Inequality}) $a_k=-\frac{u_k}{\sqrt{n}}$, $b_k=\frac{u_k}{\sqrt{n}}$, and $S_n = \sum_{k=1}^{n}X_k$, where $\textup{E}\left ( S_n \right ) = \textup{E}\left ( \sum_{k=1}^{n}X_k \right ) = \textup{E}\left [ \sum_{k=1}^{n} \left ( u_k \cdot v_{k}^{i} \right ) \right ] = 0$ by (\ref{equation 3 of northo}). Hoeffding's inequality thus gives the following results:

\begin{equation*}
\forall \varepsilon \in \mathbb{R}^{+}:\textup{P}\left [ \left | S_{n} - \textup{E}\left ( S_{n} \right ) \right | \geq   \varepsilon \right ] \leq 2\, \textup{exp}\left [ -\frac{2\varepsilon^{2}}{\sum_{k=1}^{n}\left ( b_{k}-a_{k} \right )^{2}} \right ] \: \left [ \textup{original statement of (\ref{equation 1 of northo}) before substitution} \right ]
\end{equation*}

\begin{equation*}
\textup{P}\left ( \left | \sum_{k=1}^{n} X_k-0\right | \geq \varepsilon \right ) \leq 2\, \textup{exp} \left \{ -\frac{2\varepsilon^2}{\sum_{k=1}^{n} \left [ \frac{u_k}{\sqrt{n}} - \left ( -\frac{u_k}{\sqrt{n}} \right ) \right ]^2} \right \} \: \left [ \textup{substituting into (\ref{equation 1 of northo})} \right ]
\end{equation*}

\begin{equation*}
\textup{P}\left ( \left | \sum_{k=1}^{n} X_k \right | \geq \varepsilon \right ) \leq 2 \, \textup{exp} \left [ -\frac{2\varepsilon^2}{\sum_{k=1}^{n} \left ( \frac{2u_k}{\sqrt{n}} \right )^2} \right ] = 2\, \textup{exp} \left ( -\frac{2\varepsilon^2}{\frac{4}{n}\sum_{k=1}^{n}u_{k}^{2}} \right ) = 2\, \textup{exp} \left ( \frac{-\frac{1}{2}\varepsilon^2n}{\sum_{k=1}^{n}u_{k}^{2}} \right )
= 2\, \textup{exp} \left ( \frac{-\frac{\varepsilon^2n}{2}}{\left | \vec{u} \right |^2} \right )
\end{equation*}

Recall that $\vec{u}$ is a unit vector. Then $\left | \vec{u} \right | = \left | \vec{u} \right |^2 = 1$, and we conclude that
\begin{equation}
\label{equation 4 of northo}
\textup{P}\left ( \left | \sum_{k=1}^{n} X_k \right | \geq \varepsilon \right ) \leq 2\, \textup{exp} \left ( -\frac{\varepsilon^2n}{2} \right )
\end{equation}

Now, we can check the value of the inner product of two random unit vectors $\vec{v}^j$ and $\vec{v}^i$, as follows:
\begin{equation*}
\textup{P}\left ( \forall i,j \in \mathbb{N}_{\leq m}: \left | \left \langle \vec{v}^j,\vec{v}^i \right \rangle \right | \leq \varepsilon \right ) = 1 - \textup{P}\left ( \forall i,j \in \mathbb{N}_{\leq m}: \left | \left \langle \vec{v}^j,\vec{v}^i \right \rangle \right | >  \varepsilon \right ) \geq 1 - \sum_{i<j}^{} \textup{P}\left ( \left | \left \langle \vec{v}^j,\vec{v}^i \right \rangle \right | >  \varepsilon \right )
\end{equation*}

\begin{equation*}
= 1 - \sum_{i<j}^{} \textup{P}\left ( \left | \left \langle \vec{v}^j,\vec{v}^i \right \rangle \right | \geq  \varepsilon \right ) = 1 - \binom{m}{2}\, \textup{P}\left ( \left | \left \langle \vec{v}^j,\vec{v}^i \right \rangle \right | \geq  \varepsilon \right )
\end{equation*}

\begin{equation*}
\geq 1 - m^2\, \textup{P}\left ( \left | \left \langle \vec{v}^j,\vec{v}^i \right \rangle \right | \geq  \varepsilon \right ) \: \left [ \because m > 0 \Rightarrow m^2 > \binom{m}{2} = \frac{m\left ( m-1 \right )}{2} \right ]
\end{equation*}

\begin{equation*}
= 1 - m^2\, \textup{P}\left ( \left | \left \langle \vec{u},\vec{v}^i \right \rangle \right | \geq  \varepsilon \right ) \, \left ( \because \vec{u} \textup{ is an arbitrary unit vector, we can replace } \vec{v}^j \textup{with } 
\vec{u} \right )
\end{equation*}

\begin{equation}
\label{equation 5 of northo}
= 1 - m^2\, \textup{P}\left ( \left | \sum_{k=1}^{n}X_k \right | \geq  \varepsilon \right ) \, \left [ \because \left \langle \vec{u},\vec{v}^i \right \rangle = \sum_{k=1}^{n}\left ( u_k \cdot v_k^i \right )
= \sum_{k=1}^{n}X_k \right ]
\end{equation}

At this stage, we know from (\ref{equation 4 of northo}) that $\textup{P}\left ( \left | \sum_{k=1}^{n} X_k \right | \geq \varepsilon \right ) \leq 2\, \textup{exp} \left ( -\frac{\varepsilon^2n}{2} \right )$. Of course, $m > 0 \Rightarrow -m^2 < 0$, implying that $-m^2\, \textup{P}\left ( \left | \sum_{k=1}^{n} X_k \right | \geq \varepsilon \right ) \geq -2m^{2}\, \textup{exp} \left ( -\frac{\varepsilon^2n}{2} \right )$. Therefore, $1-m^2\, \textup{P}\left ( \left | \sum_{k=1}^{n} X_k \right | \geq \varepsilon \right ) \geq 1-2m^{2}\, \textup{exp} \left ( -\frac{\varepsilon^2n}{2} \right )$. We can substitute this into (\ref{equation 5 of northo}):
\begin{equation*}
\textup{P}\left ( \forall i,j \in \mathbb{N}_{\leq m}: \left | \left \langle \vec{v}^j,\vec{v}^i \right \rangle \right | \leq \varepsilon \right ) \geq 1-m^2\, \textup{P}\left ( \left | \sum_{k=1}^{n} X_k \right | \geq \varepsilon \right ) \geq 1-2m^{2}\, \textup{exp} \left ( -\frac{\varepsilon^2n}{2} \right )
\end{equation*}

Then, we substitute in our chosen values of $\varepsilon$ and $m$.
\begin{equation*}
\textup{P}\left ( \forall i,j \in \mathbb{N}_{\leq m}: \left | \left \langle \vec{v}^j,\vec{v}^i \right \rangle \right | \leq \sqrt{\frac{5}{\sqrt{n}}} \right ) \geq 1 - 2\left ( e^{\sqrt{n}} \right )^2\, \textup{exp} \left [ -\frac{\left ( \sqrt{\frac{5}{\sqrt{n}}} \right )^2n}{2} \right ]
\end{equation*}

\begin{equation}
\label{equation 6 of northo}
= 1-2\, \textup{exp}\left ( 2\sqrt{n} \right )\, \textup{exp}\left ( -\frac{5\sqrt{n}}{2} \right ) = 1-2\, \textup{exp} \left ( -\frac{1\sqrt{n}}{2} \right )
\end{equation}

Lastly, we can finally address Proposition \ref{proposition: as n grows, inner product tends to 0} by taking $n \to \infty$ on the result in (\ref{equation 6 of northo}):
\begin{equation*}
\lim_{n\to\infty} \textup{P}\left ( \forall i,j \in \mathbb{N}_{\leq m}: \left | \left \langle \vec{v}^j,\vec{v}^i \right \rangle \right | \leq \sqrt{\frac{5}{\sqrt{n}}} \right ) \geq \lim_{n\to\infty} \left [ 1-2\, \textup{exp} \left ( -\frac{5\sqrt{n}}{2} \right ) \right ]
\end{equation*}

\begin{equation*}
\textup{P}\left ( \forall i,j \in \mathbb{N}_{\leq m}: \left | \left \langle \vec{v}^j,\vec{v}^i \right \rangle \right | \leq 0 \right ) \geq 1-2\, \lim_{n\to\infty}\, \textup{exp} \left ( -\frac{5\sqrt{n}}{2} \right )
\end{equation*}

\begin{equation*}
\textup{P}\left ( \forall i,j \in \mathbb{N}_{\leq m}: \left | \left \langle \vec{v}^j,\vec{v}^i \right \rangle \right | \leq 0 \right ) \geq 1-2(0)
\end{equation*}

\begin{equation*}
\textup{P}\left ( \forall i,j \in \mathbb{N}_{\leq m}: \left | \left \langle \vec{v}^j,\vec{v}^i \right \rangle \right | \leq 0 \right ) \geq 1
\end{equation*}

\begin{equation*}
\textup{P}\left ( \forall i,j \in \mathbb{N}_{\leq m}: \left | \left \langle \vec{v}^j,\vec{v}^i \right \rangle \right | = 0 \right )= 1 \Rightarrow \left \langle \vec{v}^j,\vec{v}^i \right \rangle \to 0
\end{equation*}

And thus, as $n \to \infty$, the probability of the inner product of two random unit vectors being zero approaches 1, and so we conclude that $n \to \infty \Rightarrow \left \langle \vec{v}^j,\vec{v}^i \right \rangle \to 0$ as desired in Proposition \ref{proposition: as n grows, inner product tends to 0}.
\end{proof}

\begin{remark}
When generating random unit vectors in $\mathbb{R}^n$, each vector component must take its value from the set $\left \{ -\frac{1}{\sqrt{n}},\frac{1}{\sqrt{n}} \right \}$. This ensures that regardless of how large $n$ is, the norm of the vector is 1, as desired. This is easy to verify:
\begin{equation*}
\left | \vec{v}^i \right |= \sqrt{\left ( \pm \frac{1}{\sqrt{n}}\right )^2n}=1 \textup{ as desired.} \: \: \square
\end{equation*}
\end{remark}

\subsubsection{Simulation to show Almost Orthogonality in Higher Dimensions}
\label{subsubsection: Simulation for norhto in high dim}
Other than a proof of the "curse of dimensionality", Proposition \ref{proposition: as n grows, inner product tends to 0} describes a method to generate random vectors. This helps us simulate what happens when dimensionality increases, then check for pairwise orthogonality via the value of the inner product, or equivalently checking for cosine differences.

Figure 1 illustrates the resulting cosine similarity value, which ranges between -1 and 1, where 1 indicates identical directions, and -1 indicates opposite directions. Thus, the desired value is 0, which indicates orthogonality, and indeed we observe that the values approach 0 rapidly as dimensionality increases.

\begin{figure}[h]
  \centering
  \includegraphics[width=\linewidth]{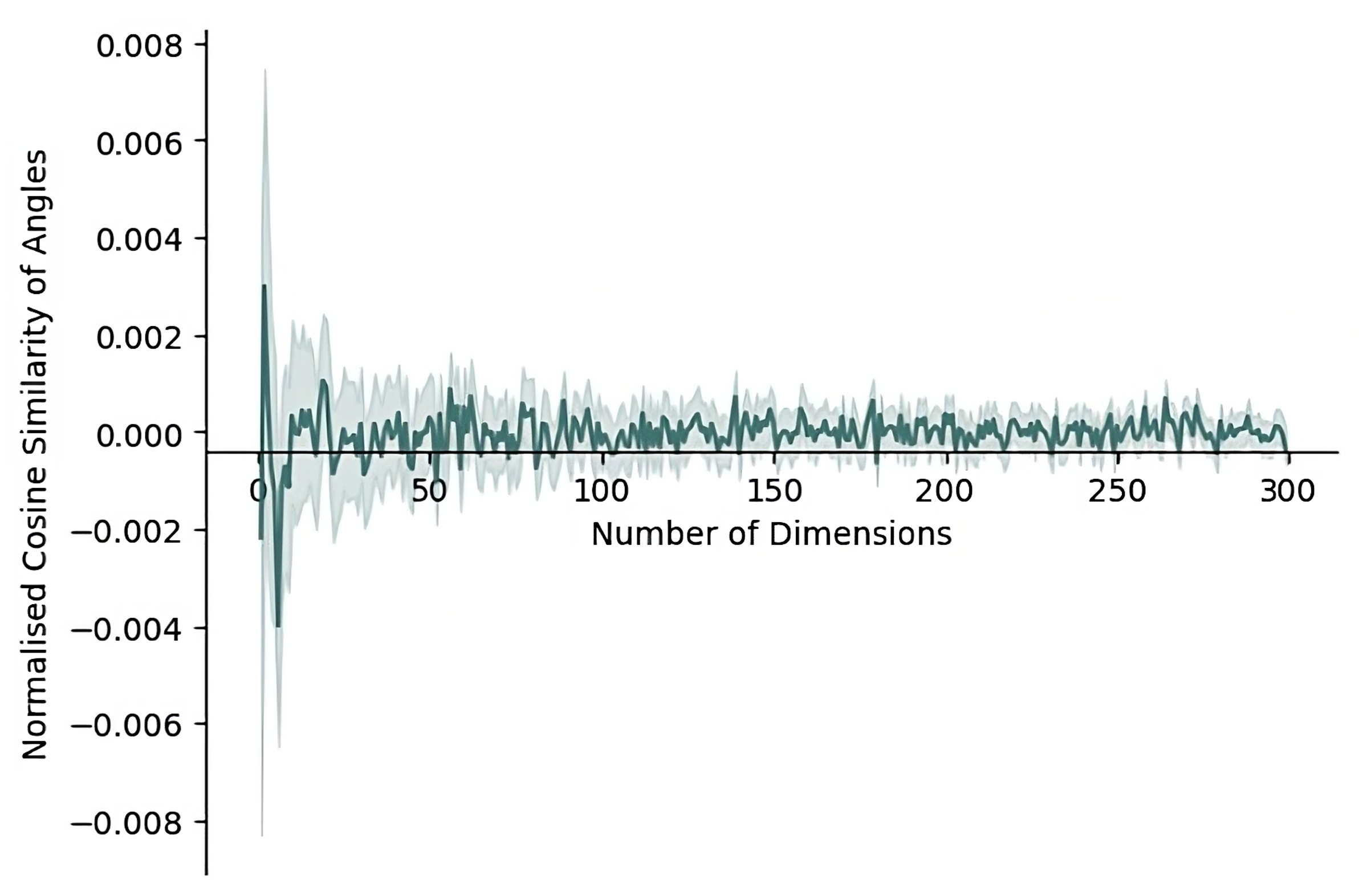}
  \caption{Graph showing Simulation of Cosine Differences between one hundred thousand vector samples per given N dimension}
  \Description{The graph shows that the cosine difference tends towards zero implying increasing occurrences of near orthogonality}
\end{figure}

\subsection{Implications of Orthogonality on Data}
High-dimensional random vectors becoming almost orthogonal has significant implications for data processing. Here are a few ways in which it affects the field. In the context of high-dimensional data, these are challenges encountered in data analysis:

\subsubsection{Predictive inaccuracy}
Orthogonal data may fail to capture interdependencies between variables that are essential for understanding underlying mechanisms. This causes an incomplete representation of the various features in data. For example in medical data, patient demographics or medical history might be incompletely represented thus, possibly generating oversimplified models with limited predictive power. Without accurate predictions, diagnosis, treatment planning, and prognosis are affected.

\subsubsection{Increased vulnerability to noise}
In high-dimensional data, noise or measurement errors are more damaging due to the large number of variables. Orthogonal data fails to capture the dependencies like described earlier, so we lose the ability to use these dependencies to distinguish meaningful signals from random fluctuations.



\subsection{Computational and Memory Constraints}

High-dimensional datasets impose formidable computational burdens, as algorithms traditionally designed for lower dimensions exhibit increased time complexity and memory requirements. Analytical procedures that thrive in low-dimensional settings struggle to navigate the vast parameter spaces inherent in high dimensions, leading to computational inefficiencies and prohibitive resource demands. The computational challenges associated with high dimensions impede the scalability of analyses, hindering the timely processing and interpretation of large-scale datasets.

\subsection{Data Sparsity and Overfitting}

The sparsity induced by the curse of dimensionality exacerbates the risk of overfitting, a phenomenon where models capture noise as if it were genuine signal. In high-dimensional spaces, the abundance of features provides more opportunities for chance correlations, amplifying the susceptibility to overfitting. As a consequence, predictive models trained on such datasets may fail to generalize well to new observations, compromising their reliability and predictive accuracy. Addressing the delicate balance between model complexity and generalizability becomes paramount in mitigating the adverse effects of sparsity and overfitting.

\subsection{Dimensionality Reduction and Advanced Analytical Approaches}

In response to the challenges posed by high-dimensional data, dimensionality reduction techniques have emerged as pivotal tools to distill essential information and mitigate the curse of dimensionality. Methods such as principal component analysis (PCA) and manifold learning aim to uncover lower-dimensional representations that preserve the salient features of the data. As we confront the complexities of high-dimensional datasets, a synthesis of these techniques becomes imperative for unlocking the latent insights embedded in the expansive realms of multidimensional data.

\section{Exploring Traditional Dimensionality Reduction Methods}

\subsection{Principal Component Analysis (PCA)}
\subsubsection{What is PCA}
Principal Component Analysis\cite{pearson1901liii} is a widely used linear dimensionality reduction technique. It simplifies high-dimensional data by transforming it into a lower-dimensional space while retaining the most significant information. This transformation achieves dimensionality reduction by identifying new axes, called principal components, that capture the largest variances in the data.

\subsubsection{Theoretical Foundations}

PCA relies on the concept of eigenvalues and eigenvectors. Eigenvalues, denoted by $\lambda_i$, represent the amount of variance explained by each principal component $i$, while eigenvectors, denoted by $u_i$, define the directions of these components in the original high-dimensional space. PCA identifies principal components in order of decreasing variance, ensuring that the first few components capture the most significant information in the data.

\subsubsection{Applications in Exploratory Data Analysis}

PCA finds numerous applications in exploratory data analysis:

\begin{itemize}
    \item Visualization: By projecting high-dimensional data onto the first few principal components, PCA allows for visualization in lower dimensions, facilitating exploration of data patterns and relationships.
    \item Feature selection: PCA identifies the most informative features by ranking them based on their associated variance. This helps select the most relevant features for further analysis, reducing computational costs and potential overfitting.
    \item Anomaly detection: Deviations from the captured variance by PCA can indicate outliers or anomalies in the data, allowing for their identification and investigation.
\end{itemize}

\subsubsection{Strengths}

\begin{itemize}
    \item Model Interpretability: Principal components represent the directions of maximum variance, providing insights into the underlying structure of the data. As a note to the reader, it is important to understand that Model Interpretability pertains to the understanding of observable manifestation of cause-and-effect relationships within a system, specifically denoting the capacity to anticipate outcomes in response to alterations in input or algorithmic parameters. Hence, when we refer to interpretability throughout this paper we are not directly addressing the output that we observed but rather regarding the "explainability" of the given output.
    \item Efficiency: PCA is computationally efficient, making it suitable for analyzing large datasets.
\end{itemize}

\subsubsection{Weaknesses}

\begin{itemize}
    \item Assumes linearity: PCA assumes a linear relationship between features, which might not hold true for complex datasets.
    \item Ignores non-linear relationships: PCA may not effectively preserve well-separated clusters in the original space if the data exhibits non-linear relationships.
\end{itemize}

PCA's ability to condense information while retaining the essential variance makes it an invaluable tool for enhancing interpretability. The principal components, being orthogonal, provide a clear view of the underlying relationships within the data. Additionally, PCA enables insightful visualizations, as the first few principal components often capture the dominant patterns, allowing for a simplified yet comprehensive representation of the dataset's structure.

\subsection{t-Distributed Stochastic Neighbor Embedding (t-SNE)}
\subsubsection{What is t-SNE}

t-Distributed Stochastic Neighbor Embedding\cite{van2008visualizing}, is a non-linear dimensionality reduction technique designed to visualize high-dimensional data by preserving the local similarities between data points. Unlike PCA, which focuses on capturing global variance, t-SNE aims to embed the data in a lower-dimensional space while maintaining the distances between neighboring points in the original space.

\subsubsection{Theoretical Foundations}

t-SNE operates by calculating pairwise similarities between data points in both the high-dimensional and low-dimensional spaces. These similarities are often represented by a Gaussian kernel in the high-dimensional space and a Student t-distribution kernel in the low-dimensional space. It then minimizes a cost function, often the Kullback-Leibler divergence, that measures the difference between these similarities, ensuring that nearby points in the original space remain close together in the lower-dimensional representation.

\subsubsection{Applications in Exploratory Data Analysis}

t-SNE finds numerous applications in exploratory data analysis for high-dimensional datasets with non-linear relationships:

\begin{itemize}
    \item Visualization: t-SNE excels at revealing non-linear structures and clusters in high-dimensional data, making it a valuable tool for understanding complex relationships and identifying hidden patterns.
    \item Clustering: By preserving local similarities, t-SNE can aid in identifying clusters and exploring their characteristics in high-dimensional data.
\end{itemize}

\subsubsection{Strengths}

\begin{itemize}
    \item Handles non-linear relationships: Effective for visualizing complex datasets with intricate structures and non-linear relationships.
    \item Preserves local similarities: Captures the relationships between neighboring data points, allowing for better cluster visualization.
\end{itemize}

\subsubsection{Weaknesses}

\begin{itemize}
    \item Less interpretable: The resulting dimensions may not have a clear physical meaning, making interpretation challenging.
    \item Computationally expensive: Can be computationally expensive for large datasets, especially compared to PCA.
\end{itemize}

t-SNE finds widespread applications in various domains, including image analysis, natural language processing, and bioinformatics. Its ability to reveal intricate structures and clusters within high-dimensional data makes it a valuable tool for exploratory data analysis and pattern recognition. t-SNE has been employed to visualize complex relationships in biological data, uncover hidden patterns in images, and assist in understanding the semantic relationships between words.

\subsection{Uniform Manifold Approximation and Projection for Dimension Reduction (UMAP)}
\subsubsection{What is UMAP}
Uniform Manifold Approximation and Projection\cite{mcinnes2018umap}, is a recent dimensionality reduction technique that extends concepts from t-SNE, utilizing a topological perspective. It focuses on preserving local structures in the data, constructing a low-dimensional representation that captures relationships between neighboring points, ultimately retaining the overall high-dimensional topological structures.

\subsubsection{Theoretical Foundations}
UMAP constructs a fuzzy topological representation of the data using nearest neighbors, capturing connectivity between data points for the preservation of local structures. Optimization involves a cost function considering topological information, encouraging proximity between points in the original and low-dimensional spaces. UMAP employs an efficient optimization strategy with adaptive learning rates, enhancing computational speed compared to t-SNE.

\subsubsection{Applications in Exploratory Data Analysis}
UMAP proves valuable for analyzing complex, high-dimensional datasets with both linear and non-linear relationships. It excels in visualization, revealing intricate cluster patterns and broader relationships within the data. UMAP aids in clustering, offering insights into both cluster characteristics and the high-dimensional data's topological structure.

\subsubsection{Strengths}
\begin{itemize}
    \item Handles non-linear relationships effectively.
    \item Preserves both topological structures, capturing fundamental relationships and connections between data points.
    \item Relatively interpretable, providing insights into the data structure compared to t-SNE.
\end{itemize}

\subsubsection{Weaknesses}
\begin{itemize}
    \item Less interpretable than PCA, with resulting dimensions lacking a clear physical meaning.
    \item Can be computationally expensive for large datasets, though generally faster than t-SNE.
\end{itemize}

UMAP excels at visualizing intricate structures and hidden patterns in complex, high-dimensional data, even when those relationships are non-linear, aiding it's efficacy in various applications, such as discerning concealed patterns in images and attaining a holistic comprehension of intricate relationships in biological data. This renders UMAP as a powerhouse for delivering visual representations conducive to in-depth exploratory analyses.

\section{Proposed Methodology}

While Principal Component Analysis (PCA) effectively reduces dimensionality by leveraging principal components (eigenvalues), it can struggle with capturing non-linear relationships. Conversely, techniques like UMAP and t-SNE, which utilize local embedding, excel at preserving local structure but may sacrifice global information, particularly for smaller datasets.

Our proposed method addresses this trade-off by seeking to capture both interpretable global structure, similar to PCA, and also metric information embedded in the data, akin to approaches like t-SNE and UMAP. The primary objective is to enable faster visualization, in particular, for high dimensional datasets characterized by small class sizes and low sample sizes.

\subsection{Key steps in our method:}

\begin{enumerate}
    \item Convex Hull Generation: We begin by constructing the convex envelope of the discrete point cloud data. This convex shape encapsulates the entire data distribution.
    \item Bounding Ball Construction: We then create a bounding ball around the convex hull. The surface of this ball, the hypersphere of which will serve as a pseudo-ground truth representation space for our classes.
    \item Characteristic High Dimension Surface Class Solver: For every data point, we engage in a process to project instances associated with same class label onto the surface of the sphere. This involves infimising over a chosen distance metric, to find a point on the sphere, for each class, that has the smallest distance to all the points of a given class. This process subsequently dictates the precise spatial coordinates of each class on the spherical surface. 
    \item Similarity Matrix Calculation: We then compute the distance from each data point to every designated point representing the Characteristic High Dimension Surface Class on the surface of the sphere, employing a selected distance metric. This distance serves as an indicator of the similarity between the data point and each class. Consequently, this will elucidate the extent to which each data point embodies the characteristics of each class.
    \item Visualization: Following the construction of a regular simple polygon with the Characteristic High Dimension Surface Class as vertices, the next step involves mapping the points from the high dimension to two dimensions. Each data point is positioned within the polygon based on weighted distances to the class points obtained from the previously calculated Similarity Matrix. These weights are determined through simple convex combination calculations. This mapping process ensures that each data point is accurately placed within the two-dimensional space, taking into account its proximity to the class points and the associated weights obtained from the Similarity Matrix.
\end{enumerate}

\begin{figure}[H]
  \centering
  \includegraphics[width=0.93\textwidth,height=\textheight,keepaspectratio]{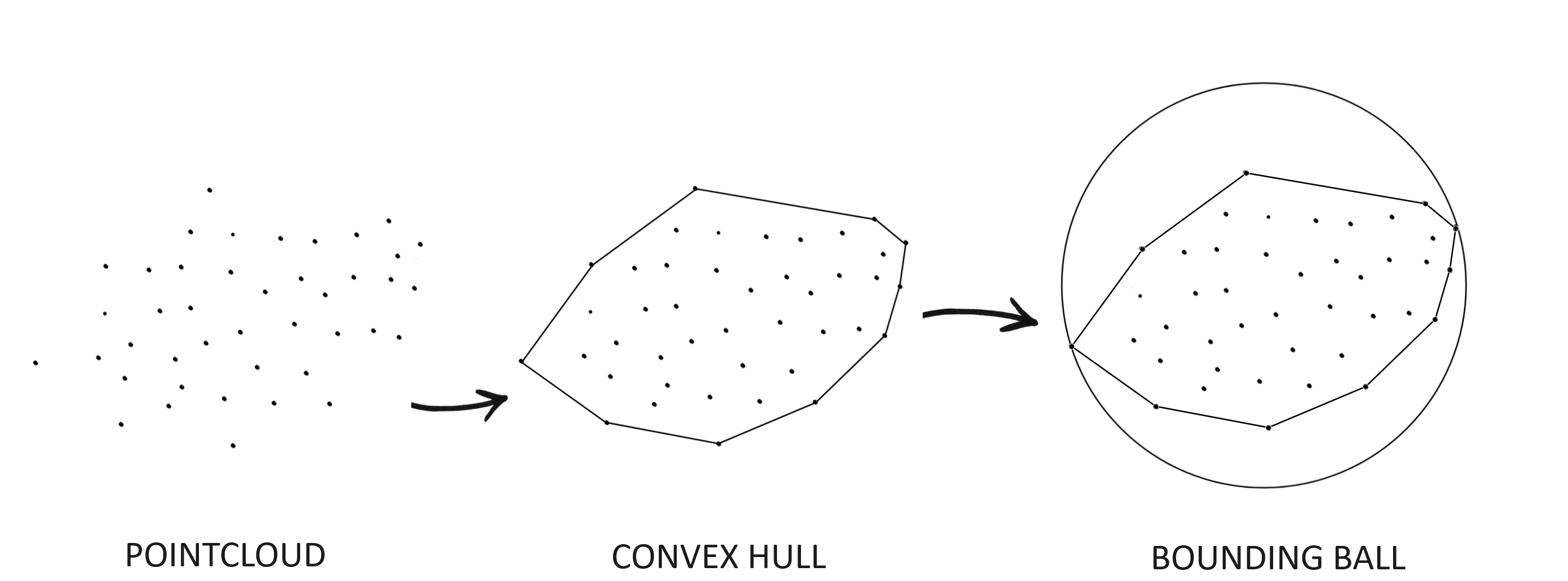}
  \caption{Pictorial Representation of Convex Hull and Bounding Ball over Data Points}
  \Description{Picture Showing Convex Hull and Bounding Ball}
  \label{ConvexHullBoundingBall-VIS}
\end{figure}

This approach offers several advantages:

\begin{itemize}

     \item Preservation of Global Structure and Handling Non-linear Relationships: Leveraging the convex hull and its bounding sphere and mapping into a regular polygon sustains an understanding of the broader data distribution, akin to PCA. In contrast to PCA's assumption of linearity, our method adeptly manages non-linear data structures by leveraging the distance metric employed to project data points onto the sphere's surface. This capability ensures the preservation of valuable non-linear relationships among data points via metric information, aligning with techniques like t-SNE and UMAP.
    
    \item Interpretable Visualisation: By mapping data points onto a pre-defined geometric shape (sphere) and utilizing class-specific points as anchors in another geometric shape (regular simple polygon), the visualization offers a more intuitive interpretation compared to the abstract lower-dimensional spaces generated by PCA and t-SNE. This can be particularly valuable when exploring data with well-defined class structures.

    \item Rapid Acquisition of Relationships from the Visualisation: This method explicitly incorporates class labels throughout the process, designating points as "Characteristic High Dimension Surface Class" to represent each class later within the visualization space. This approach facilitates a class-centric organization of the data, potentially enhancing separation and interpretability for class-related analysis. By capitalizing on the pre-defined structure of the polygon outlined by these classes, our method not only allows for efficient visualization but is also poised to be particularly advantageous for quickly understanding relationships present in smaller datasets.
\end{itemize}

\subsection{Underlying Assumptions:}

Our method operates under the assumption that each data point is expressible as an interpolation among various classes, suggesting the presence of convexity within the dataset allowing for the existence of intermediate points between classes.

Additionally, we posit an assumption regarding the homology of the overarching data distribution, asserting its homeomorphism to a "genus 0" (indicating the absence of holes) n-ball without internal cavities. This topological consideration ensures the continuity of the transformation mapping to the regular polygon, validating its applicability in our approach.

Furthermore, there is a reliance on having pre-defined class labels for data points. This limits its applicability to scenarios where such labels are not readily available or where the data structure itself lacks well-defined class boundaries.

\subsection{Rationale for Bounding Ball}

The significance of Riemannian manifolds in higher-dimensional setting arises when exploring and attempting to port over and extend familiar analytical tools traditionally employed in lower-dimensional settings. Mathematicians, in this pursuit, usually seek to make measurements in these high dimensional spaces and in this context, a crucial distinction emerges between a general manifold of higher dimension and a Riemannian manifold. Unlike a general manifold, the Riemannian manifold has a positive definite inner product defined on the tangent plane. This characteristic ensures that the conventional understanding of measurement, as applicable in the familiar one, two, and three dimensions setting, persists in the realm of high dimensional Riemannian manifolds.

Another supporting decision to incorporate a spherical envelope in our methodology stems from its ability to preserve the sense of perfect "mid-point of all classes" interpolation while facilitating an intuitive mapping onto a simple regular two dimensional polygon. At the heart of this choice lies the concept that a point located at the center of the sphere, when transported, maintains its status as the perfect mid-point, ensuring consistency in representation. We prove this at the end of the following section. This property aligns seamlessly with the underlying objective of creating visualizations that reflect the intrinsic relationships within high-dimensional data.

The geometric harmony offered by the sphere contributes to the uniform weightage and transport of points across, resulting in a more spatially balanced representation. Leveraging the sphere's properties like radius and centroid during our measures of distances ensures not only a consistent spatial relationship and increased speed of computations in higher dimensions but also enhances the interpretability of the resulting 2D visualization we aim to observe. In essence, the choice of a spherical envelope provides a more visually coherent and intuitively interpretable framework

\section{Theoretical Foundations for SPREV}

\subsection{Formalising the proposed approach}

\begin{definition} \textbf{Supervised Learning Dataset.}
In the context of Supervised Learning, consider a dataset \( \mathcal{D} \) comprising \( m \) datapoints, each characterized by a feature vector of size \( n \). This dataset can be formally represented as:
\[ \mathcal{D} = \{(x_i, y_i) \mid \forall i \in \mathbb{N}, \exists m \in \mathbb{N}, 1 \leq i \leq m\} \quad \Rightarrow\{(x_1, y_1), (x_2, y_2), \ldots, (x_m, y_m)\}\]
Here, \( x_i \) denotes the feature vector of the \(i\)-th datapoint, and \( y_i \) represents the corresponding output label or target variable. This representation encapsulates the input-output pairs within the dataset for all \( i \) belonging to the integer values in the range \( 1 \leq i \leq m \) in the supervised learning context.

\end{definition}

\begin{definition} \textbf{Class Labels.}
In the context of our task, we interpret the outputs \( y_i \) in \( \mathcal{D} \) as class labels belonging to a set \( C \), where each \( y_i \) corresponds to some \( c_i \). Formally, \( C = \{c_1, c_2, \ldots, c_q\}\), and the output-label relationship is defined by a mapping \[ f: Y \rightarrow C \quad \Rightarrow  \quad f(y_i) = c_j, \quad m \leq q, \quad 1 \leq i \leq m, \quad 1 \leq j \leq q, \quad \forall i, j \in \mathbb{N}  , \quad  \exists m, q \in \mathbb{N} \]
This mapping guarantees that each class label in \( C \) possesses at least one corresponding output label in \( Y \), and conversely, every element in \( Y \) maps to \( C \). The cardinality of such a class set \( C \) is at most as large as our output set, \( Y \), a reflection of practical considerations when establishing clustering groups in supervised learning to discern patterns. An alternative perspective on this procedure is as the establishment of class boundaries as a means to enhance the comprehension of the underlying data's structure in the context of our problem.
\end{definition}

\begin{corollary}
From the aforementioned definitions, it follows that the dataset \(\mathcal{D}\), originally represented as \[\mathcal{D} = \{(x_i, y_i) \mid \forall i \in \mathbb{N}, \exists m \in \mathbb{N} , 1 \leq i \leq m\}\] can alternatively be expressed as \[\mathcal{D} =  \{ (x_i, c_j) \mid \forall i, j\in \mathbb{N},\exists m,q \in \mathbb{N}, 1 \leq i \leq m, 1 \leq j \leq q, q\leq m\}\]

In this alternative representation, \(x_i\) denotes the feature vector of the \(i\)-th datapoint, and \(c_j\) signifies the \(j\)-th distinct class label. This formulation will now serve as the foundation for elucidating the subsequent steps in our process.

\end{corollary}

\subsection{Convex Hull Generation Step}

\begin{definition} \textbf{Convex Hull}

The convex hull of a set \(X\) is defined as the smallest convex set that contains \(X\). Considering a finite set of points \(x_1, x_2, \ldots, x_n\), the convex hull, denoted as \(\text{conv}(X)\), is defined as:
\[ \text{conv}(X) = \left\{ \sum_{i=1}^{n} \lambda_i x_i \mid \exists n \in \mathbb{N}, x_i \in X,   \sum_{i=1}^{n} \lambda_i = 1, \lambda_i \geq 0 , \forall i \in \mathbb{N} , 1\leq i \leq n \right\} \]

In this definition:
\begin{itemize}
  \item \(x_i\) are the finite countable elements (points) in the set \(X\).
  \item \(n\) denotes the cardinality of the set \(X\), and thus signifies the number of elements, or points, participating in the convex combination.
  \item \(\lambda_i\) are non-negative real number coefficients, representing the weights of the points in the convex combination.
  \item The sum of the weights \(\sum_{i=1}^{n} \lambda_i\) equals 1, ensuring that it is a convex combination.
\end{itemize}

Geometrically, the convex hull is the smallest convex shape that encloses all points in the set \(X\), and any point within the convex hull can be expressed as a convex combination of points from \(X\).
\end{definition}
\begin{corollary}
If \(Y\) is any convex set that contains \(X\), then \(\text{conv}(X) \subseteq Y\)
\end{corollary}

\begin{definition}
Let \( \Delta_{(n-1)} \) denote the standard (n-1)-simplex in \( \mathbb{R}^n \), defined as the set of n standard unit vectors(points) \(s_1, \ldots, s_n\) such that \[\Delta_{(n-1)}  = \left\{ (s_1, \ldots, s_n) \in \mathbb{R}^n : s_i \geq 0,  \sum_{i=1}^{n} s_i = 1, \forall i \in \mathbb{N} , 1 \leq i \leq n, \exists n \in \mathbb{N} \}\right\} \].
\end{definition}

\begin{theorem}
\( \Delta_{(n-1)} \) is a compact set
\end{theorem}

\begin{proof}
Consider the following hyperplane:
\[ H = \left\{ (s_1, \ldots, s_n) \in \mathbb{R}^n \mid \sum_{i=1}^{n} s_i = 1 , \forall i \in \mathbb{N} , 1 \leq i \leq n, \exists n \in \mathbb{N} \}\right\}  \]

By definition, \( \mathbb{R}^k_+ \) is closed, and \( H \) is also closed. The set \( \Delta_{(n-1)}\) can be re-expressed as the intersection of these closed sets: \( \Delta^n = \mathbb{R}^k_+ \cap H \).
Therefore, \( \Delta_{(n-1)}\) is also closed by De Morgan's theorem where an intersection of arbitrary number of closed sets is itself closed.\cite{rudin1976principles}

Consider the hypercube:
\[[0,1]^n\]

Given that \( \Delta_{(n-1)}\) the subset of the hypercube\cite{1807022}, it is therefore bounded.

\end{proof}

\begin{definition}
A vector space can be defined as a two-tuple \[(V, \mathbb{F})\] where \cite{strang2019introduction}
\cite{axler2015linear}

\begin{itemize}
    \item $u,v,w \in V$ , a set of elements called vectors
    \item $\mathbb{F}$ is a field (like  the real numbers $\mathbb{R}$ or complex numbers $\mathbb{C}$) whose elements $a,b \in \mathbb{F}$ are numbers that follow standard algebra rules\cite{4682000}. We refer to these elements as scalars.
    \item $+$ is an operation on $V$ called \textbf{vector addition} satisfying the following 4 axioms:
    \begin{itemize}
        \item \textbf{Commutativity}: $u + v = v + u$ for all $u, v \in V$.
        \item \textbf{Associativity}: $(u + v) + w = u + (v + w)$ for all $u, v, w \in V$.
        \item \textbf{Identity element (Existence of a zero vector)}: There exists a unique element $0 \in V$ such that $u + 0 = u$ for all $u \in V$.
        \item \textbf{Inverse elements (Existence of additive inverses)}: For every $u \in V$, there exists a unique element $-u \in V$ such that $u + (-u) = 0$.
    \end{itemize}
\item $\cdot$ is a \textbf{scalar multiplication} operation from $F \times V$ to $V$, satisfying the following 4 axioms:
    \begin{itemize}
        \item \textbf{Distributivity with respect to field addition}: $(a + b) \cdot u = a \cdot u + b \cdot u$ for all $a, b \in F$ and $u \in V$.
        \item \textbf{Compatibility with field multiplication}: $a \cdot (b \cdot u) = (ab) \cdot u$ for all $a, b \in F$ and $u \in V$.
        \item \textbf{Distributivity with respect to vector addition}: $a \cdot (u + v) = a \cdot u + a \cdot v$ for all $a \in F$ and $u, v \in V$.
        \item \textbf{Identity element}: $1 \cdot u = u$ for all $u \in V$, where 1 is the multiplicative identity (an element leaves unchanged every element when the operation is applied\cite{weissteinIdentity}) in $F$.
    \end{itemize}
    \item Vector addition: $\mathbf{u} + \mathbf{v} = \mathbf{w}$ $\implies$ $\mathbf{u}, \mathbf{v}, \mathbf{w} \in V$.
    \item Scalar multiplication: $a \cdot \mathbf{u} = \mathbf{v}$ $\implies$ $a \in F$ and $\mathbf{v}, \mathbf{u} \in V$.

\end{itemize}

\end{definition}

\begin{definition}
A \textbf{finite-dimensional vector space} \cite{halmos1974finite} is one where there exists a finite subset $\{v_1, v_2, \dots, v_n\} \subset V$ such that:

\begin{itemize}
    \item \textbf{Spanning}: Every vector $u \in V$ can be expressed as a linear combination of the basis vectors: $u = c_1 v_1 + c_2 v_2 + \dots + c_n v_n$ for some scalars $c_1, c_2, \dots, c_n \in F$.
    \item \textbf{Linear independence}: No non-trivial linear combination of the basis vectors is equal to the zero vector: $c_1 v_1 + c_2 v_2 + \dots + c_n v_n = 0$ implies $c_1 = c_2 = \dots = c_n = 0$.
\end{itemize}

The smallest such set satisfying both properties is called a \textbf{basis}, and its cardinality is the \textbf{dimension} of the space, denoted by $\dim(V)$.

Simply put, a vector space is called finite dimensional if it is spanned by a finite set of vectors.

\end{definition}

\begin{theorem} \textbf{All norms are equivalent in a finite dimensional space.}
Let \(V\) be a real or complex finite-dimensional vector space, and let \(\| \cdot \|_a\) and \(\| \cdot \|_b\) be norms on \(V\). Then \(\| \cdot \|_a\) and \(\| \cdot \|_b\) are equivalent, in the sense that there exist strictly positive constants \(C, D > 0\) such that for every vector \(v \in V\),
\[ C \|v\|_b \leq \|v\|_a \leq D \|v\|_b. \]

This can be proved via Heine–Borel as done by Steven G. Johnson \cite{johnson2012notes} or alternatively we can use the triangle inequality and Bolzano–Weierstrass theorem to construct an inequality to come to a similar conclusion that can be found in standard functional analysis texts.\cite{fitzpatrick2009advanced}
\end{theorem}

\begin{theorem} 
The convex hull of a finite set is compact.
\end{theorem}
\begin{proof}

Let \(\{x_1, \ldots, x_n\} \in X\) be a finite subset of some Banach space \(B\). Considering the standard simplex in \(\mathbb{R}^n\), and the following map
\[ \phi: (s_1, \ldots, s_n) \in \Delta_{(n-1)} \mapsto (s_1x_1 + \ldots + s_nx_n) \in X \subset B, \]

\begin{itemize}
    \item Since \textbf{addition} is \textbf{continuous}, (because of triangle inequality)
    \item and \textbf{multiplication} with scalars is \textbf{continuous} due to positive homogeneity of the norm (WLOG, we can obtain this from the preceding theorem, given norm equivalence in finite dimensional space) \cite{1509256}
\end{itemize}

Hence, $\phi$ (which is a composition of the above two continuous functions) is also \textbf{continuous}, and its image is \(\text{conv}\{x_1, \ldots, x_n\}\). $\because$ \(S\) is compact $\Rightarrow$ \(\text{conv}\{X\}\) is compact.\cite{22608}

\end{proof}

\begin{corollary} \textbf{Convex Envelope of Dataset, $\mathcal{M}^n$, is compact.}
For a given dataset \( \mathcal{D} \), consisting of \( m \) discrete points in a space of finite dimension \( n \), the convex hull, \( \mathcal{M}^n \) that encapsulates \( \mathcal{D} \) in the space is the smallest convex set that contains it. From the preceeding theorem, we can easily see that this convex set is compact.
\end{corollary}

\subsection{Bounding Ball Generation Step}

\begin{corollary} \textbf{Existance of $n$-ball Encapsulation of Dataset.}
For a given compact convex hull \( \mathcal{M}^n \) of dimension \( n \), there exists a hypersphere \( B^n \) that completely covers and encapsulates the entire convex hull.
\end{corollary}

\begin{proof}
Let $E$ be a \textbf{finite-dimensional} vector space, we denote the closed ball with center, $c$ and radius $r$ by:

\[ {B}(c, r) = \left\{ y \in E \mid \|y - c\| \leq r , \quad \exists r > 0 ,\quad \exists c \in E \right\}  \]

$\mathcal{M}^n \ \subset E $ where $\mathcal{M}^n \neq \emptyset$ and $\mathcal{M}^n$ is compact.

Since $\mathcal{M}^n$ is compact and hence, bounded, then the set 
\[\bar{R} = \{r \in \mathbb{R}_{+}^{*} \mid \exists c \in E, \quad \mathcal{M}^n  \subset {B}(c, r)\}\] 

is not empty, and we can consider $\rho = \inf(\bar{R})$.

For every $i \in \mathbb{N}^{*}$, there exist $\rho_i \in [\rho, \rho + \frac{1}{i}]$ and $c_i \in E$ such that $\mathcal{M}^n  \subset {B}(c_i, \rho_i)$.

Fix $m \in \mathcal{M}^n$. For all $i \ge 1$, we have:
$$
\Vert x_i \Vert - \Vert m \Vert \le \Vert c_i - m \Vert \le \rho_i \le \rho + \frac{1}{i} \le \rho + 1.
$$

So, the sequence $(c_i)$ is bounded. Since $\dim(E) = n < \infty$, there exists a convergent subsequence $(c_{\varphi(i)})$.

Set $a = \lim_{i\to\infty}c_{\varphi(i)}$. Taking the limit ($i \to \infty$) in $\Vert c_{\varphi(i)} - m \Vert \le \rho_i$, we get $\Vert a - m \Vert \le \rho$\cite{2100328}

$\Rightarrow$ $\mathcal{M}^n \subset {B}(a, \rho)$.
\end{proof}

\begin{definition} \textbf{Characteristic High Dimension Surface Class.} 
The boundary of the n-ball, we derived earlier, is defined by the hypersphere:

\[S(c, r) = \left \{ y \in E \mid \|y - c\| = r, \quad \exists r > 0 ,\quad \exists c \in E \right\}\]

where c is the center and r is the radius.

To find each Characteristic High Dimension Surface Class, we do the following:

\[P_\mathbf{C} = \inf_{\substack{z \in S(c,r)\\ \forall x \in \mathbf{C}}} ||z-x||\]

where x are the points in the dataset that belong to class $\mathbf{C}$ and z are the points that lie on the surface of the hypersphere.
\end{definition}

\subsection{Visualisation Step}

\begin{definition} \textbf{Characteristic $\mathbb{R}^2$ Vertex Class.}
Given the cardinality of C to be $n_c$, then the vertex class are defined as the vertices of a \textbf{regular simple} $n_c$-gon to which we do a one-to-one mapping of each $P_c$ to. 
\end{definition}

\subsection{Convexity Property Preservation (Preserving the existence of an "Interpolation Space")}

\begin{definition}
A regular simple polygon in $\mathbb{R}^2$, is characterized by the simultaneous fulfillment of two distinctive properties: 

\begin{enumerate}
\item \textbf{Equilaterality}, signifying uniformity in the lengths of all its sides.s
\item \textbf{Direct Equiangularity}, denoting the equality in the measurement of all internal angles
\end{enumerate}
\end{definition}

\begin{corollary}
A regular simple polygon in $\mathbb{R}^2$ as defined by the above two properties, inherently possesses convexity.    
\end{corollary}

\begin{proof}
\begin{lemma}
For any simple polygon, for any vertex in a polygon, the sum of its exterior angle (\(\alpha\)) and corresponding interior angle (\(\beta\)) at the $i$-th vertex are given by:
\[ \alpha_i + \beta_i  = \pi \quad \forall i \]

where $\pi$ radians equal to $180^\circ$ 
\end{lemma}

\begin{lemma}
When considering only one of the two external angles at each vertex, for any simple polygon, the sum of its exterior angles, represented by $\alpha$, is given by:
\[ \sum_{i=1}^{n} \alpha_i = 2\pi , \quad n \in \mathbb{N}, n\geq 2 \]
Here, $n$ is the number of vertices in the polygon, and $\alpha_i$ is the external angle at the $i$-th vertex. The sum of all external angles, $\alpha_i$, equals $2\pi$ radians ($360^\circ$)\cite{posamentier2012secrets}

\end{lemma}
In the spirit of the colloquial expression of "walking around in circle(s)", an intuitive way to think about this involves envisioning a scenario where traversing the perimeter of a simple polygon in $\mathcal{R}^2$ eventually returns one to the initial point they started at. This behavior stems from a simple polygon being characterized as a closed, non-self-intersecting continuous curve. In an alternate interpretation, it is considered homeomorphic to a circle without twisting.

\begin{lemma}
A simple polygon is termed a convex polygon if all of its interior angles, represented by \(\beta\), satisfy 

\[\beta_i \leq  \pi \ \quad \forall i \] 

and strictly convex if 

\[\beta_i <  \pi \ \quad \forall i\] 
\end{lemma}

Reformulating the first lemma above would yield the following:

\begin{equation}
\sum_{i=1}^{n} (\alpha_i + \beta_i) = n\pi  
\end{equation}

and subtracting this from the second would result in

\begin{equation}
\sum_{i=1}^{n} \beta_i = (n-2)\pi  
\end{equation}

As a consequence of Direct Equiangularity, this means that for a regular simple polygon in $\mathbb{R}^2$, this would yield 

\begin{equation}
\beta_i = \frac{(n-2)}{n}\pi , \quad  \forall i 
\end{equation}

Finally considering the third lemma, 

\begin{equation}
\beta_i = \frac{(n-2)}{n}\pi \leq \pi , \quad  \forall i  \quad \because  n \geq 2 , n \in \mathbb{N}
\end{equation}

Implying convexity as desired.
\end{proof}

\subsection{Centroid Distance Preservation (Preserving "Global Centering")}

\begin{corollary}
In a regular simple polygon in $\mathbb{R}^2$ the centeroid \(R_c\) is equidistant from all its vertices. 
\end{corollary}
\emph{This is in line with the idea that the surface points on the hypersphere are equidistant from the n-ball's centroid.}
\begin{proof}

\begin{lemma}
    A general regular polygon (whether convex or star) is a cyclic polygon.
\end{lemma}

A direct consequence of this lemma is that all vertices of the regular simple polygon ($\subset$ general regular polygon) lie on the circumscribed circle. 

Given a circle, $\bar{C}$ with center c and radius r can be described as the following:

\[\bar{C}(c,r) = \left\{ y \in \mathbb{R}_2 \mid \|y - c\| = r, \quad \exists r > 0 ,\quad \exists c \in {R}_2 \right\}\] 

Therefore, demonstrating since all the vertices have to respect the above constraints as they lie on  $\bar{C} \Rightarrow$ that the centeroid \(R_c\) = c is equidistant, precisely r distance away, from all the vertices. The corollary is thus, proved.

\end{proof}

\section{Computational Considerations and Adjustments for SPREV}

In the realm of computational algorithms, the transition from theoretical frameworks to practical implementation often necessitates a distinct perspective and approach. This imperative arises from the inherent dichotomy between continuous theoretical constructs and the discrete nature of computational settings. As computation unfolds within discrete domains, the seamless integration of theoretical frameworks demands thoughtful rephrasing and adaptation. Particularly, concepts originally formulated in a "continuous" setting must undergo transformation to align with the discrete nature of computational processes. This essential reinterpretation not only ensures the feasibility of computable frameworks but also paves the way for computational speedups. Thus, acknowledging the imperative for a nuanced and tailored computational viewpoint becomes paramount in bridging the theoretical-practical gap through some assumptions and refactoring to unleash the mathematical frameworks suggested in the earlier for use in a discrete computational domain.

\subsection{Core Algorithm}
The algorithm commences by initializing the samples and classes $\mathcal{X}$, $\mathcal{C}$, from dataset, $\mathcal{D}$. Subsequently, Min-Max Scaling is applied to $\mathcal{X}$, generating a normalized dataset $\mathcal{X}_t$, in a bid to decrease time spent computing the convex hull by limiting the data to a unit hypercube. Class centroids $\mathcal{X}_c$ are then computed based on $\mathcal{X}_t$ and class information $\mathcal{C}$. Utilizing a specialized function, the algorithm identifies a set $\mathcal{P}$ representing surface class points from the hypersphere the circumscribes the unit hypercube. A similarity matrix $\mathcal{M}_s$ is generated by evaluating the similarity between $\mathcal{X}_t$ and $\mathcal{P}$ through distance measurement. Finally, the algorithm employs a plotting function to visualize the resultant similarity matrix $\mathcal{M}_s$.

\begin{algorithm}[H]
\caption{SPREV algorithm}
\label{SPREV:main}
\SetKwProg{SPREV}{Function \emph{SPREV}}{}{end}

\KwIn{Supervised Learning Dataset}
\KwOut{ Visualisation Plot}
\BlankLine
\textbf{Data}: $\mathcal{D}$\;
\SPREV{($\mathcal{D}$)}{
     $\mathcal{X}$,$\mathcal{C}$ =  $\mathcal{D}$\;
     $\mathcal{X}_t$ = \emph{MinMaxScaling}($\mathcal{X}$)\;
     $\mathcal{X}_c$ = \emph{ClassCentroids}($\mathcal{X}_t$,$\mathcal{C}$)\;
     $\mathcal{P}$ = \emph{SurfaceClass}($\mathcal{X}_c$)\;
     $\mathcal{M}_s$ = \emph{GenerateSimilarityMatrix}($\mathcal{X}_t$, $\mathcal{P}$)\;
     \emph{Plot}($\mathcal{M}_s$)\;
}
\end{algorithm}

\subsection{Convex Hull Alternative}
To find a convex hull in $\mathbb{R}_2$ we can reach out to Grahams Scan\cite{graham1972efficient}, a viable O(NlogN) algorithm. However, extending this method to three dimensions and beyond poses challenges. This is because a surface homeomorphic to a sphere lacks a natural linear ordering of its points. The determination of the next point to add or the decision of whether to accept a point and proceed in a "left" or "right" direction becomes ambiguous in three-dimensional space. The absence of semantic coherence, particularly the challenge of defining "left" and "right" in higher-dimensional space compared to the 2-dimensional implementation, makes it difficult to generalize the original algorithm into higher dimensions.

An alternative approach for finding a convex hull in $\mathbb{R}_2$ is the gift wrapping algorithm, also known as Jarvis march\cite{jarvis1973identification}. This algorithm has a time complexity of O(nh), where n is the number of points and h is the number of points on the convex hull. Its real-life performance, when compared with other convex hull algorithms, is advantageous when dealing with small sample sizes (which aligns with our use case) or when h is expected to be very small relative to n, which we cannot guarantee for our algorithm.
We still would required a smallest bounding sphere algorithm after this and for high dimensions this is 

Since finding a convex hull and bounding ball in high dimensions is an expensive operation, we can reach out for an alternative method which is to min-max scale the values and this would result in them being bounded by the unit hypercube as the convex hull after which the proceeding bounding ball can be quickly found by calculating the principal diagonal of this hypercube which will yield the diameter of the hypersphere that circumscribes this hypercube.

\begin{theorem}
To enhance the computational efficiency and simplify our algorithm, consider normalizing the dataset, envisioning the data enclosed within a hypercube. Consequently, our problem transforms into determining the covering sphere of said hypercube. A plausible strategy involves the following:

Perform min-max scaling on the features of our dataset, mapping values from the interval [min, max] to [0, 1]. This transformation is achieved by subtracting the minimum value from each data point and subsequently dividing by the range (max - min). This process ensures that all feature values are distributed within the range [0, 1].

This normalization results in the encapsulation of the dataset within the unit hypercube, serving as the convex hull, provided there are at least two distinct points in the dataset, each represented by all feature values.
\end{theorem}
\begin{proof}

The viability of this method hinges upon the existence of minimum and maximum values within the dataset. Through this procedure, these extremal values are systematically transformed to 0 and 1, respectively. Consequently, the corners of the hypercube coincide with the points mapped to these extreme values. By considering the union of all simplices formed by these corner points, it becomes evident that this ensemble constitutes the unit hypercube. It is straightforward to observe that any additional simplex, formed by points not located at the corners, is inherently a subset of the union that forms this hypercube. Consequently, we establish that the hypercube serves as the convex hull, defined as the union of all simplices with vertices within the hypercube.

\end{proof}

\begin{algorithm}[H]
\SetKwProg{MinMaxScaling}{Function \emph{MinMaxScaling}}{}{end}
    \caption{MinMaxScaling Algorithm}
    \label{algo:minmaxscaling}
    \KwIn{Input data}
    \KwOut{Scaled data}

    \BlankLine
\MinMaxScaling{($\mathcal{X}$)}{
    $min\_data \leftarrow \min(\mathcal{X})$ \;$\quad$ $\quad$ \tcp{Find the minimum value in the dataset}
    $max\_data \leftarrow \max(\mathcal{X})$\;$\quad$ $\quad$ \tcp{Find the maximum value in the dataset}
    
    \For{$i \leftarrow 1$ \KwTo $|\mathcal{X}|$}{
        $scaled\_data[i] \leftarrow \frac{(\mathcal{X}[i] - min\_data) }{(max\_data - min\_data)} $ \;$\quad$ $\quad$ \tcp{Perform Min-Max Scaling}
    }
    
    \KwRet $scaled\_data$\;
    }

\end{algorithm}

\subsection{Infimum Alternative}
Infimising distances, especially in high-dimensional spaces, can impose a significant computational burden and demand substantial resources. The complexity arises from the requirement to compute distances between a point located on the surface and all other relevant class points within the dataset. This challenge escalates notably as the dataset expands in size. Considering the necessity to perform this calculation for every potential point on the sphere's surface, which represents the brute force search space, renders the approach computationally impractical.

Hence, we would have to reach out for a heuristic that would be sufficient and practical for our use case to make this algorithm practical. Hence, finding the class centroid from "averaging" all the relevant class points and projecting a ray through the center of the n-ball to the surface of the sphere offers a pragmatic and computationally efficient approximation alternative. This approach's calculation involves determining the central point of a class,a ray cast and a surface collision, which significantly reduces the computational burden compared to computing distances between an infinite pair of points.

The following is the algorithm to calculate the class centroids:

\begin{algorithm}[H]
\SetKwProg{ClassCentroids}{Function \emph{ClassCentroids}}{}{end}
    \caption{Class Centroids Algorithm}
    \label{algo:classcentr}
    \KwIn{Samples and their corresponding Classes}
    \KwOut{Class Centroids Matrix}

    \BlankLine
\ClassCentroids{($\mathcal{X}_t$,$\mathcal{C}$)}{

\For{each class $(c_i) \in \mathcal{C}$}{
      \For{each vector $(x_i) \in \mathcal{X}_t$}{
        \For{each projection ($x_{i,j}$) in dimension $j$ of $x_i$}{
      
            $\text{total}[c_i][j] \leftarrow \text{total}[c_i][j] + x_{i,j}$\;
  }}
  $\mathcal{X}_c[c_i] \leftarrow \frac{\text{sum}[c_i]}{size(total[c_i])}$\;
  }

    \KwRet $\mathcal{X}_c$\;
    }

\end{algorithm}

The ray casting and collision calculations can be simplified into a concise vector calculation when armed with information about the ball's center coordinates and its radius. By calculating the direction vector through a straightforward computation involving the coordinates of the center of the ball and the class centroid, we can then leverage the radius length to ascertain the precise intersection point where the ray meets the surface of the sphere.

The preceding selection of employing the unit hypercube as the convex hull turns out to be tremendously advantageous here. By opting for the circumscribing hypersphere of this unit hypecube as our bounding ball, we can analytically derive the center and radius of the ball, resulting in a substantial speed up the computation to find the surface class of interest.

\begin{theorem}
The diameter of the n-sphere that circumscribing the n-cube, which is equal to the prinicpal diagonal of the n-cube is $\sqrt{n} R$ . 
\end{theorem}
\begin{proof}

Using the metric of $n$-dimensional Euclidean space $g_{ab} = \delta_{ab}$ for $a,b \in [1,2,\dots,n]$. Let $s$, the distance be then defined as
$$
(\mathrm{d} s)^2 =  \sum_{a,b} g_{ab} \mathrm{d}x^a \mathrm{d}x^b \, ,
$$
Let the curve $x^a = x^a(t)$ parameterised by $t$. Then
$$
\frac{\mathrm{d} s}{\mathrm{d} t} = \sqrt{\sum_{a,b} g_{ab} \frac{\mathrm{d}x^a}{\mathrm{d}t} \frac{\mathrm{d}x^b}{\mathrm{d}t}} = \sqrt{\sum_a \left ( \frac{\mathrm{d}x^a}{\mathrm{d}t} \right ) ^2 } \, .
$$
The diagonal going from $(0,0,\dots,0)$ to $(R,R,\dots,R)$ can be described by the curve $x^a(t) = Rt$ for $t \in [0,1]$. The total length of the curve is
$$
s = \int_0^1 \mathrm{d} s = \sqrt{\sum_a \left ( \frac{\mathrm{d}x^a}{\mathrm{d}t} \right ) ^2 } \mathrm{d}t = \int_0^1 \sqrt{\sum_a \left ( R \right ) ^2 } \mathrm{d}t = \int_0^1 \sqrt{n} R \mathrm{d} t = \sqrt{n} R \, .  
$$
Therefore the length of the principal diagonal in $n$ dimensions is $\sqrt{n} R$. \cite{1603122}

This principal diagonal constitutes the circumradius under consideration. As inferred from the preceding derivation, given the uniform length of all principal diagonals, we can construct a covering n-sphere, granted the principal diagonal representing the greatest straight line length encompassed within the cube as per definition.

\end{proof}

\begin{algorithm}[H]
\caption{Surface Class Algorithm}
\label{SPREV:Surfaceclass}
\SetKwProg{SurfaceClass}{Function \emph{SurfaceClass}}{}{end}

    \KwIn{Samples and their corresponding Classes}
    \KwOut{Surface points corresponding to Class}

    \BlankLine

\textbf{Samples Class Centroid}: $\mathcal{X}_c$\;
\textbf{Number of dimensions of X in $\mathcal{D}$}: ${n}$\;
\textbf{Principal Radius of Unit Hypercube}: ${R = 1}$\;
\textbf{Centroid of n-ball}: ${B_c}$\;
\SurfaceClass{($\mathcal{X}_c$)}{
    nballRadius = $\sqrt{n} \cdot R$ \;$\quad$ $\quad$ \tcp{ $\sqrt{n} $ because of sphere covering unit hypercube proof from earlier}
     $D$ = $\mathcal{X}_c$ - ${B_c}$\; $\quad$ $\quad$ \tcp{direction vectors}
     $D_n$ = $\frac{D}{\lVert D \rVert}$\; $\quad$ $\quad$ \tcp{normalised direction vectors}
     $\mathcal{P}$ = ${B_c}$ + $R \cdot $ $D_n$ \;
     \textbf{return} $\mathcal{P}$\;
}
\end{algorithm}

\subsection{Similarity Generalised as Distance}

The application of distance measurements to evaluate similarity constitutes a fundamental aspect in various disciplines. Distance metrics, such as Euclidean distance, Cosine similarity, and the Wasserstein distance, offer quantitative measures to discern the proximity or dissimilarity between data points.  Researchers face the crucial task of selecting the most pertinent distance metric for their specific use case, a decision contingent upon the characteristics of the data and the analytical objectives. By generalizing similarity through the lens of various distance metrics, we introduce a modular element to this approach. This modularity enables adaptability, allowing researchers or analysts to tailor their methodology according to the specific goals of their analysis, thereby enhancing precision and facilitating applicability across diverse scientific and business domains.

\begin{algorithm}[H]
\SetKwProg{GenerateSimilarityMatrix}{Function \emph{GenerateSimilarityMatrix}}{}{end}
    \caption{Similarity Matrix Algorithm}
    \label{algo:SimMat}
    \KwIn{Samples and Surface Classes}
    \KwOut{Similarity Matrix}

    \BlankLine
\GenerateSimilarityMatrix{($\mathcal{X}_t$,$\mathcal{P}$)}{

\For{each class $(c) \in \mathcal{P}$}{
\For{each vector $(x_i) \in \mathcal{X}_t$}{
$\mathcal{M}_s[i][c]$ = dist($c$,$x_i$)\; $\quad$ $\quad$ \tcp{dist is the function to calculate the distance between the two vectors for a given metric}
}
}

    \KwRet $\mathcal{M}_s$\;
    }

\end{algorithm}

\subsection{Plotting Algorithms}
The algorithms presented herein pertain to the generation of visualizations. The lineplot and scatterplot algorithms are designed to serve as generic implementations, ensuring ease of comprehension and implementation for users in whichever programming language or framework of their choice. The foundational concept involves the systematic generation of a base simple regular polygon. By leveraging similarity information generate in the Similarity Matrix Algorithm and coordinates of the vertices of a simple regular polygon, the algorithms then proceed to construct a scatterplot encapsulated in the shape. In this scatterplot, the information from our dataset is represented in a manner that aligns with the similarity to the surface class points, allowing for a meaningful visualization of the data points within the confines of the generated polygon.

\begin{algorithm}[H]
\SetKwProg{Plot}{Function \emph{Plot}}{}{end}
  \caption{Plot Algorithm}
  \label{algo:Plot}

  \KwIn{Distance Matrix}
  \KwOut{SPREV Plot}
  \BlankLine
\Plot{($\mathcal{M}_s$)}{
  
  $coords =$ \emph{PlotPolygon}($\mathcal{M}_s$)\;
    $convexMatrix =$ \emph{convexComb}($\mathcal{M}_s$, $coords$)\;
    \emph{scatterplot}($convexMatrix$)\;

  \Return \;
}
\end{algorithm}

\begin{algorithm}[H]
\SetKwProg{PlotPolygon}{Function \emph{PlotPolygon}}{}{end}
  \caption{Plot Polygon Algorithm}
  \label{algo:PlotPoly}

  \KwIn{$\mathcal{M}_s$}
  \KwOut{Plot of regular-n-Polygon, Coordinate list of vertices}
  \BlankLine
\PlotPolygon{($\mathcal{M}_s$)}{
  num = number of columns of $\mathcal{M}_s$\;
  $\theta =$ \emph{LinSpace}(0,$2 \cdot \pi$,num)\;
    $x$ = cos($\theta$)\;
    $y$ = sin($\theta$)\;
    coords = [$x$,$y$]\;
    \emph{lineplot}(coords)\;
    \Return coords

}
\end{algorithm}

\begin{algorithm}[H]
\SetKwProg{LinSpace}{Function \emph{LinSpace}}{}{end}
  \caption{Linear Space Algorithm}
  \label{algo:linspace}

  \KwIn{start, stop, num}
  \KwOut{Array of evenly spaced values}
  \BlankLine
\LinSpace{(start,stop,num)}{
  
  $step \leftarrow \frac{{stop - start}}{{num - 1}}$ \;$\quad$ $\quad$\tcp{Calculate step size}

  \For{$i \leftarrow 0$ \KwTo $(num-1)$}{
    $result[i] \leftarrow start + i \cdot step$ \;
  }

  \Return{$result$} \;
}
\end{algorithm}

\begin{algorithm}[H]
\SetKwProg{convexComb}{Function \emph{convexComb}}{}{end}
  \caption{Convex Combination Algorithm}
  \label{algo:convexComb}

  \KwIn{Distance Matrix of size (M x N), Vertice Coordinates of Regular-N-Polygon of size (2xN)}
  \KwOut{Convex Combination Matrix of size (Mx2)}
  \BlankLine
  
\convexComb{($\mathcal{M}_s$, $coords$)}{
  
  $convexCombMtx \leftarrow \mathcal{M}_s \cdot (coords)^\top $ \;$\quad$ $\quad$\tcp{Linear Algebra Matrix Dot Product}   

  \Return convexCombMtx\;
}
\end{algorithm}

The systematic and detailed procedures embedded in these algorithms are deliberately structured to enhance reproducibility and foster accessibility for a diverse range of users. By providing a step-by-step framework, our objective is to not only facilitate the straightforward replication of results but also to encourage users to adopt and adapt these algorithms to their specific analytical contexts. This approach aligns with the principles of transparency and methodological clarity in scientific research, ensuring that the intricacies of the visualization process are comprehensible and accessible to a broader audience. The intention is to empower researchers, practitioners, and data analysts to seamlessly integrate these algorithms into their own workflows, thereby contributing to the broader dissemination and applicability of the proposed methodologies.

\section{Exploring the Data Sets}
\subsection{MNIST Dataset}
The MNIST\cite{MNISTcite} dataset, a benchmark in the field of machine learning and computer vision, has played a pivotal role in shaping the landscape of image recognition algorithms\cite{baldominos2019survey}.

\subsubsection{Dataset Information}
The MNIST dataset, comprising a vast collection of 28x28 grayscale handwritten digits, has served as a cornerstone for benchmarking machine learning algorithms, particularly in the realm of image classification. Originating from the National Institute of Standards and Technology (NIST), MNIST represents a diverse set of handwritten characters, encompassing digits from 0 to 9. This dataset's popularity stems from its accessibility, simplicity, and the ability to benchmark a wide array of classification models.

\subsubsection{Dataset Characteristics}
MNIST consists of 60,000 training images and 10,000 test images, each labeled with the corresponding digit it represents. The images, initially introduced to facilitate the development of automatic digit recognition systems, encapsulate variations in writing styles, stroke thickness, and orientation. The balanced distribution of digits ensures a representative sample for training and evaluation purposes.

\subsubsection{Reformulation into Tabular Form}
The 28x28 images are reformulated into tabular form by computing each pixel value, resulting in a dataset of 784 dimensions and 70,000 samples with 10 classes. This sets the dataset as one of medium class size, high dimension, high sample size. We will use this as a yardstick to compare with the small class size, high dimension, low sample size culled variant of MNIST that we will generate to allow us scrutinize the performance of the dimension reduction algorithms under these different situations. The method of which we will generate the culled variant will be elaborated upon in the final subsection of this section.

\subsubsection{Significance related to our Research Focus}
The capacity to express numerical values through distinct written representations that convey identical semantic meanings aligns with the argument of convexity. This argument posits the existence of a conceptual space wherein semantic significance remains invariant by variations in stylistic expression.

\subsubsection{Closing remarks}The dataset's simplicity, counterposed against the intricate challenges it presents, establishes MNIST as a foundational benchmark in the realm of image classification for researchers and practitioners. This attribute positions MNIST as an essential tool in our testing endeavors, aiming to align with standardized evaluation methods allows for our research to be both explainable and amenable to reproducibility.

\subsection{Fashion-MNIST Dataset}
The Fashion-MNIST\cite{xiao2017fashion} dataset, akin to its predecessor MNIST, has emerged as a pivotal asset in the domain of machine learning and computer vision, significantly influencing the development of image recognition algorithms\cite{leithardt2021classifying}\cite{meshkini2020analysis}.

\subsubsection{Dataset Information}
Fashion-MNIST, an ensemble of 28x28 grayscale images, diversifies the benchmarking landscape forimage classification algorithms. Conceived as an alternative to the traditional MNIST, it originates from the need to expand benchmarking horizons beyond handwritten digits. Fashion-MNIST encompasses a variety of fashion items from Zalando's article images, providing a rich dataset for evaluating classification models.

\subsubsection{Dataset Characteristics}
With 60,000 training images and 10,000 test images, Fashion-MNIST mirrors the balanced structure of its predecessor. Each image is labeled with the corresponding fashion category it represents. The dataset captures nuances in fashion items, including variations in texture, shape, and style. The equal distribution of categories ensures a representative sample for robust training and evaluation.

\subsubsection{Reformulation into Tabular Form}
Similar to MNIST, Fashion-MNIST's 28x28 images are reformulated into tabular form via each pixel getting allocated a unique pixel value, representing the luminance of the pixel. Higher numerical values indicate decreased intensity, with the scale ranging from 0 for white to 255 for black. Hence, resulting in a dataset of 784 dimensions and 70,000 samples with 10 classes. This establishes the dataset with medium class size, high dimension, and high sample size. Our intention is to utilize this as a baseline for comparison, contrasting it with a culled variant of Fashion-MNIST in the subsequent subsection. As mentioned previously, the generation methodology for the culled variant will be elucidated in the final subsection of this section.

\subsubsection{Significance related to our Research Focus}
Fashion-MNIST dataset depicts a diverse array of fashion items based on distinct visual characteristics. This may introduce a scenario where an intermediate state between some classes of clothing articles may not inherently exist within the given dimension, implying a lack of convexity.  This has the potential to significantly challenge the genus of 0 assumption we have established, providing an opportunity to assess the robustness of our algorithm under conditions where assumptions may be breached.

To illustrate, envision the $n$-dimensional space in which the topology might delineate a clear distinction between a distinct pair of pants and a shirt, with no straightforward linear(straight-line) homotopy connecting them. This implies that within this dimension, there might not exist some intermediary variants bridging the gap between a shirt and pants. While the possibility of all intermediaries existing could emerge in a higher dimension, our restriction to the dimension of "$n$" leads to an undesirable void/gap/hole that obstructs a seamless connection between pants and shirts.
\subsubsection{Closing Remarks}
Fashion-MNIST, with its amalgamation of simplicity and intricacy, establishes itself as a fundamental benchmark for image classification research. Its role as a touchstone aligns with our research objectives, where standardized evaluation methods contribute to the elucidation and reproducibility of our findings in the field of computer vision and machine learning.

\subsection{COIL-20 Dataset}
The COIL-20\cite{nene1996columbia} dataset, a prominent resource in the domain of machine learning and computer vision, contributes significantly to the advancement of algorithms geared toward object recognition\cite{matas2000object}.

\subsubsection{Dataset Information}
The COIL-20 dataset encompasses a diverse collection of object images, specifically focusing on 20 distinct objects. Widely recognized for its role in benchmarking machine learning algorithms, COIL-20 provides a platform for evaluating models designed for object classification. Originating from comprehensive efforts in the field, this dataset serves as an invaluable asset for assessing the capabilities of image recognition systems.

\subsubsection{Reformulation into Tabular Form}
In line with preprocessing conventions, the images in COIL-20 can be reformulated into tabular form, typically involving pixel value computation. The resulting dataset may exhibit dimensions dependent on the image resolution and other factors. This will yield a dataset comprising 1440 samples, characterized by 1024 dimensions and distributed across 20 distinct classes. Consequently, it aligns with the classification of a medium class size, high dimension dataset. While the sample size can be considered relatively low, it does not meet the criterion of having fewer samples than dimensions, as per the definition of low sample size. Therefore, we will proceed to generate a culled variant for further analysis.

\subsubsection{Dataset Challenges and Nuances}
While COIL-20 offers a rich set of object images, challenges persist, such as variations in object appearance due to changes in orientation. Addressing these challenges requires robust recognition algorithms capable of handling the inherent complexities introduced by the dataset's characteristics.

\subsubsection{Significance related to our Research Focus}
COIL-20 stands out due to its emphasis on capturing a variety of object orientations, introducing a challenge where many images feature structurally similar objects. This structural similarity implies the presence of repeated patterns within the samples, necessitating the algorithm's ability to recognize and delineate them accurately during dimensionality reduction for proper classification. This unique characteristic poses an intriguing challenge, which is why it has been selected for our experimentation.

\subsubsection{Closing Remarks}
COIL-20, with its emphasis on object orientation, stands as a pivotal benchmark for algorithms targeting object recognition. Its nuanced characteristics provide a testing ground that extends beyond traditional datasets, fostering a deeper understanding of  underlying structures and complexities associated with dimension reduction on real-world object recognition scenarios.

\subsection{CIFAR-100 Dataset}

The CIFAR-100\cite{CIFAR100cite} dataset, a prominent entity within the realm of machine learning and computer vision, represents a crucial resource that has significantly influenced the advancement of image recognition algorithms\cite{sharma2018analysis}.

\subsubsection{Dataset Information}
CIFAR-100, comprising 32x32 color images, expands the horizons of image classification benchmarking. Serving as an extension of the original CIFAR-10, this dataset was conceived to diversify the evaluation criteria beyond the confines of general object recognition. CIFAR-100 encompasses a diverse set of 100 classes, each containing images representative of various fine-grained categories, thereby providing a comprehensive dataset for assessing classification models.

\subsubsection{Dataset Characteristics}
Featuring 50,000 training images and 10,000 test images, CIFAR-100 maintains a balanced structure akin to its predecessor. Each image is labeled with its corresponding fine-grained class, capturing intricate details in object representation, such as color, texture, and shape. The equitable distribution of classes ensures a representative and diverse sample for robust training and evaluation.

\subsubsection{Reformulation into Tabular Form}
CIFAR-100's 32x32 color images are transformed into tabular format, assigning unique pixel values to each pixel, representing the color intensity. The numerical scale ranges from 0 for complete absence of color to 255 for maximum color intensity for red, green and blue values. As a result, the dataset is configured with 3,072 dimensions. With a total of 60,000 samples and 100 unique classes, it manifests as a dataset marked by high class size, high dimensionality, and a substantial sample size.

\subsubsection{Significance related to our Research Focus}
The CIFAR-100 dataset portrays a wide spectrum of fine-grained object categories prompting an examination into the potential ramifications of its extensive class diversity on our method which is initially designed for small class sizes. By investigating the breakdown in performance under such large class size conditions, valuable insights may be gleaned, shedding light on the effects and challenges introduced by the substantial number of classes within the dataset.

\subsubsection{Closing Remarks}
CIFAR-100's diverse object categories challenge methods tailored for smaller classes, offering an opportunity to explore algorithmic adaptability in complex scenarios. Our engagement with CIFAR-100 advances insights in fine-grained image classification, impacting the realms of computer vision and machine learning.

\subsection{Culled Datasets for Small Class Size, High Dimensionality, and Low Sample Size Representation}

In the pursuit of addressing the challenges associated with small class size, high dimensionality, and low sample size scenarios, the construction of purposeful datasets plays a pivotal role. This section elucidates the methodology employed in the creation of culled datasets, designed to encapsulate these intricate problem settings.

The datasets utilized in our study are curated through a meticulous culling process, aiming to simulate real-world scenarios where the dimensions of the data space are high, the sample size is limited due to budget constrains or limited exploration in early stages, and classes exhibit a small representation due to focus on small set of core characteristics of interest. The culling process is orchestrated as follows:
\subsubsection{Random Class Selection}
From the original dataset, a subset of classes is randomly selected without replacement. This emulates situations where only a handful of classes are observed in practice. The flexibility to specify the number of classes allows for tailoring the dataset to varying degrees of class scarcity.

\subsubsection{Filtering and Subsampling}
Subsequently, the dataset is filtered to retain only instances corresponding to the selected classes. A fraction of the filtered dataset is then subsampled, addressing the challenge of low sample size. A "subsample fraction" parameter provides control over the ratio of retained data, offering versatility in experimenting with different levels of data sparsity.

\subsubsection{Seed for Reproducibility}
In adherence to the principles of methodological transparency and reproducibility, a predetermined random seed value is implemented in our approach, introducing a controlled source of randomness. This deliberate choice ensures the consistency of outcomes in the stochastic processes inherent to the culling method, thereby facilitating the replication of our experiments with reliability across successive iterations.

\subsubsection{Rationale for Culled Datasets}
The motivation behind the creation of these culled datasets is rooted in the necessity to mimic real-world scenarios characterized by high-dimensional feature spaces, limited availability of samples, and imbalances in class representation. By intentionally inducing these challenges, our culled datasets become invaluable tools for evaluating the robustness and generalization capabilities of algorithms in the face of data scarcity and class imbalances.

\subsubsection{Applications in Experimental Design}
The culled datasets serve as fundamental components in our experimental design, allowing us to systematically investigate the impact of small class size, high dimensionality, and low sample size. Through rigorous evaluation on these representative datasets, we aim to contribute insights that transcend the idiosyncrasies of synthetic benchmarks, providing a more realistic appraisal of algorithmic efficacy in challenging real-world conditions.

\section{Practical Efficacy}

\subsection{Qualitative aspects of Data Visualizations}
  
Quantitative metrics often play a crucial role in evaluations. However, in the realm of data visualizations, these quantiative metrics often overlook crucial aspects of human perception and interpretation. This section looks into the qualitative aspects by drawing upon insights from some commonly used data visualisations reference texts and paradigms.

Handbook of Data Visualization\cite{chen2007handbook} provides a robust foundation for understanding the theoretical underpinnings of visual perception and design principles. It emphasizes the importance of considering perceptual limitations and cognitive processes when designing effective visualizations. This aligns with Gestalt theory, which posits that humans perceive visual elements as holistic wholes rather than isolated parts\cite{wertheimer1923laws}. Principles like proximity, similarity, closure, and continuity, guide our visual organization and interpretation. By adhering to these principles, data visualizations can leverage preattentive processing, allowing viewers to readily perceive patterns and trends\cite{ware2004information}.

Tufte's design principles\cite{tufte1983visual}\cite{tufte1990envisioning}\cite{tufte2001visual} emphasize maximizing data-ink ratio, avoiding chartjunk, and prioritizing clear communication above ornamentation. The Handbook further echoes these principles, highlighting their impact on information density, visual clarity, and data integrity. These principles promote clarity, efficiency, and integrity in data presentation, enhancing viewers' ability to extract meaning from the visualization.

\subsection{Guiding Principles Behind Construction of Visualisations}
Synthesizing these insights, we generated the visualisations in line with the following distilled standards:

\subsubsection{Perception-based criteria (inspired by Gestalt theory and Handbook of Data Visualization):}
\begin{itemize}
    \item Preattentive Processing: Visual elements belonging together should be grouped using proximity, similarity, and continuity. This facilitates pattern recognition without conscious effort.
    \item Visual Hierarchy: Prioritize key information using size, color, and placement to guide viewers' attention effectively.
\end{itemize}

\subsubsection{Communication-focused criteria (inspired by Tufte's design principles):}    

\begin{itemize}
    \item Data-ink Ratio: Maximize the proportion of ink used to represent data compared to decorative elements.
    \item Chartjunk Avoidance: Eliminate unnecessary visual elements that distract from the data itself.
    \item Data Integrity: Avoid misleading or manipulative design choices that distort the true nature of the data
\end{itemize}

This proposed criteria can be systematically employed to enhance a data visualization's alignment, ensuring each criterion contributes to the overall effectiveness of the visualization. Furthermore, taking into account the specific context and carefully weighing the importance of each criterion based on the intended audience can further refine the visualization process. It is important to acknowledge, akin to Arrow's impossibility theorem\cite{kelly2014arrow}, that achieving complete satisfaction of all factors may be unattainable, thus requiring strategic maximization based on prioritized considerations.

\subsection{Qualitative Comparison of Multiple Algorithms}

In our analysis of the visualizations presented in Figures \ref{MNIST-VIS}, \ref{FMNIST-VIS}, \ref{COIL20-VIS}, and \ref{CIFAR100-VIS}, a nuanced understanding of the performance of different dimensionality reduction techniques emerges. Notably, t-SNE and UMAP consistently exhibit remarkable clustering characteristics, portraying well-defined and distinguishable groupings of data points. Conversely, PCA's clustering performance is noticeably suboptimal across most scenarios, except for Figure \ref{COIL20-VIS}, where it manages to achieve commendable separation. Our proposed method, SPREV, demonstrates a clustering quality that, while not reaching the pristine levels of t-SNE and UMAP, is still commendable.

Upon closer scrutiny of SPREV's behavior in non-culled datasets, encompassing MNIST, Fashion-MNIST, COIL-20, and CIFAR-100, a concentration of measure phenomenon becomes conspicuous. This phenomenon materializes as a substantial portion of data points tends to cluster towards the center as the class size increases. As elucidated in Section \ref{section:Near-Orthogonality}, this behavior aligns with our expectations, given the underlying utilization of an n-ball. The prominence of the concentration of measure increases with the incorporation of additional classes on the surface, leading to probabilistically smaller net differences in similarity scores measured. Consequently, the visualization encounters challenges in separation in lower dimensions due to increasingly similar similarities from every comparison with the surface, as dictated by our algorithm.

Despite the challenges posed by this phenomenon, SPREV retains a distinctive advantage in facilitating the quick visual discernment of similarities. By assessing the proximity of a class to any vertex, one can promptly gauge its similarity, leveraging the class label represented by the coloration of the corresponding point. This attribute sets SPREV apart from UMAP and t-SNE variants, as these methods do not inherently capture such global relationships in their implementations. This unique strength positions SPREV as a valuable tool for exploratory data analysis, particularly in scenarios where a rapid and intuitive understanding of global structures is essential.

Especially, in the culled variants characterized by low sample size and low class size, the utility and visual performance of SPREV becomes apparent. It facilitates swift analysis of data points, enabling the rapid assessment of inter-class relationships and identifying potential outliers within a given class.

Analyzing Figure \ref{CIFAR100-VIS} within the context of CIFAR-100 visualizations reveals that non-culled scenarios present a formidable challenge due to the dataset's intricacy, characterized by a myriad of classes. Representing such multifaceted data within a two-dimensional space inherently introduces complexity, yielding a visual output that may resemble noise, thereby complicating the identification of meaningful patterns. Despite this challenge, an intriguing observation surfaces—there exists a congruence between the global structures retained by UMAP and PCA, adding a layer of complexity to the dataset's comprehension. In contrast, t-SNE excels in generating distinct clusters, effectively mitigating some of the challenges posed by the dataset's intricate nature. 

However, we note similarly, SPREV manifests a pronounced center concentration reminiscent of PCA, while appropriately dispersing points towards the vertices. it is imperative to consider the possibility that this center clustering in UMAP, PCA and SPREV is a byproduct of the phenomenon outlined in Section \ref{section:Near-Orthogonality}. The sheer class size may result in a concentration of measure-like effect, leading to a global-scale phenomenon we observe. 

Transitioning to the culled variant of CIFAR-100, a clearer and more structured representation emerges among the selected classes. However, a salient feature persists—the prevalence of a center clustering property. Across all algorithms, a discernible center pull is observed, accompanied by a light scatter, indicative of some underlying superclass similarities among the classes. This phenomenon can likely be explained by the dataset's contextual structure, where objects in CIFAR-100 have both fine classes (which was used in labels supplied to the algorithms) and a superclass (that was not supplied to the algorithms) that encompasses these fine classes. The concerted effort by UMAP, t-SNE, and SPREV to delineate distinct classes attests to their efficacy in capturing nuanced relationships within the data, and the center mixing observed offers profound insights into what we also assume is their ability to extract the hierarchical superclass feature of classes that latently exist in the CIFAR-100 dataset that these algorithms were not exposed to.

\afterpage{
  \clearpage   

\begin{figure}[htb]
  \centering
  \includegraphics[width=0.93\textwidth,height=\textheight,keepaspectratio]{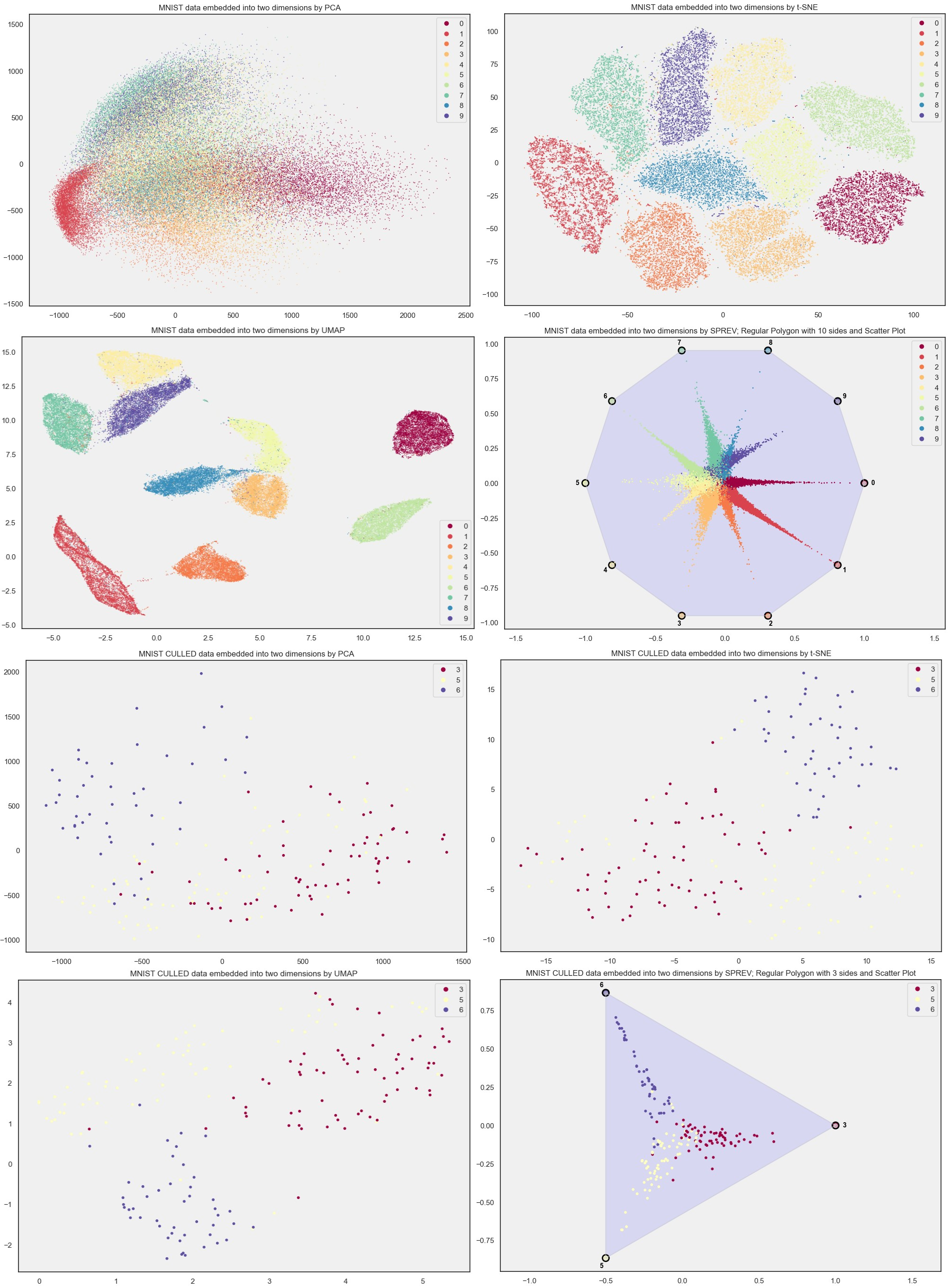}
  \caption{VISUALISATIONS of MNIST and MNIST CULLED embeddings for each of PCA, t-SNE, UMAP and SPREV.}
  \Description{Visualisation of embedding spaces for MNIST and MNIST CULLED}
  \label{MNIST-VIS}
\end{figure}

  \clearpage 
} 

\afterpage{
  \clearpage   

\begin{figure}[htb]
  \centering
  \includegraphics[width=0.93\textwidth,height=\textheight,keepaspectratio]{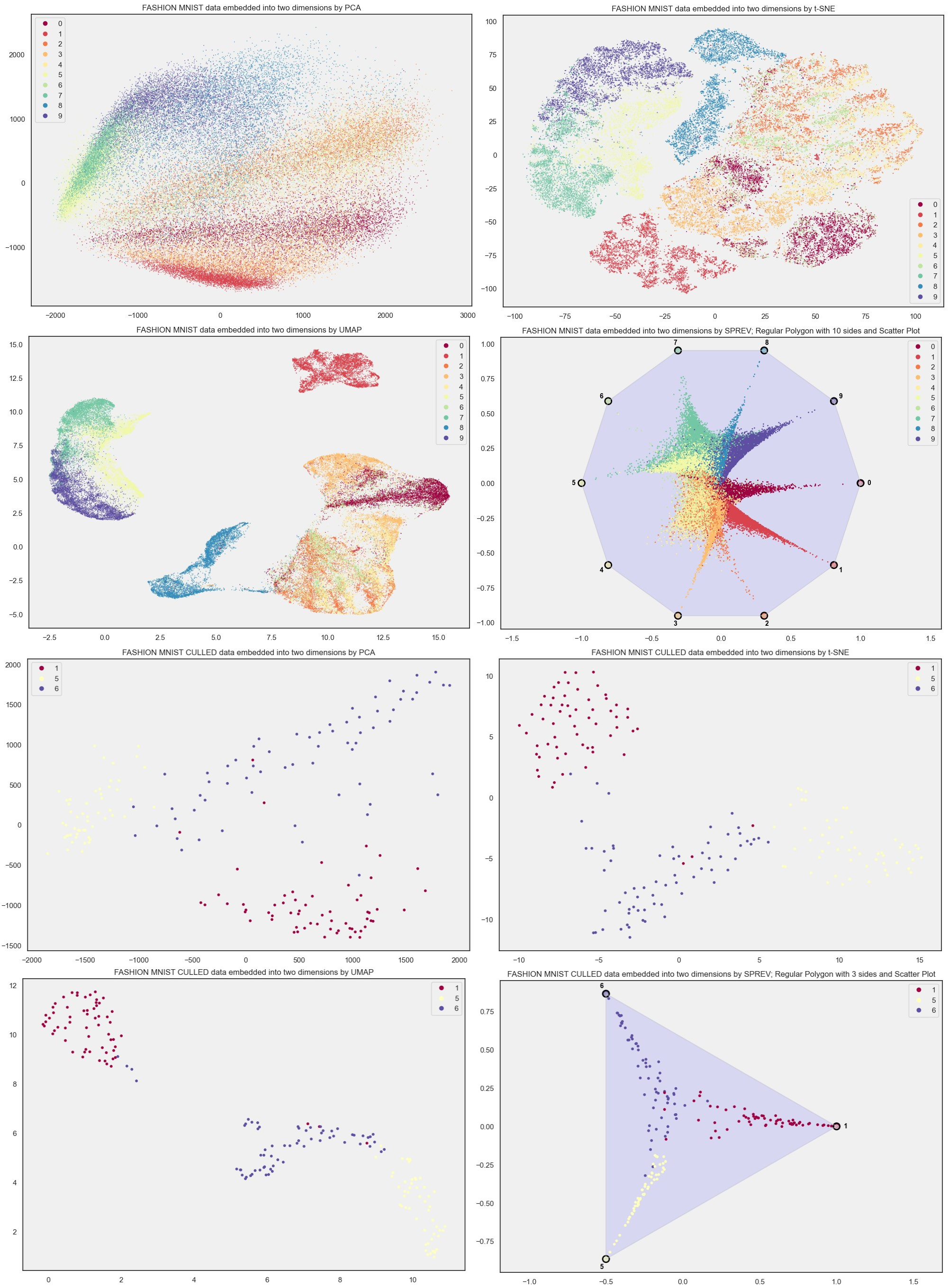}
  \caption{VISUALISATIONS of Fashion-MNIST and Fashion-MNIST CULLED embeddings for each of PCA, t-SNE, UMAP and SPREV.}
  \Description{Visualisation of embedding spaces for Fashion-MNIST and Fashion-MNIST CULLED}
  \label{FMNIST-VIS}
\end{figure}

  \clearpage 
} 

\afterpage{
  \clearpage   

\begin{figure}[htb]
  \centering
  \includegraphics[width=0.93\textwidth,height=\textheight,keepaspectratio]{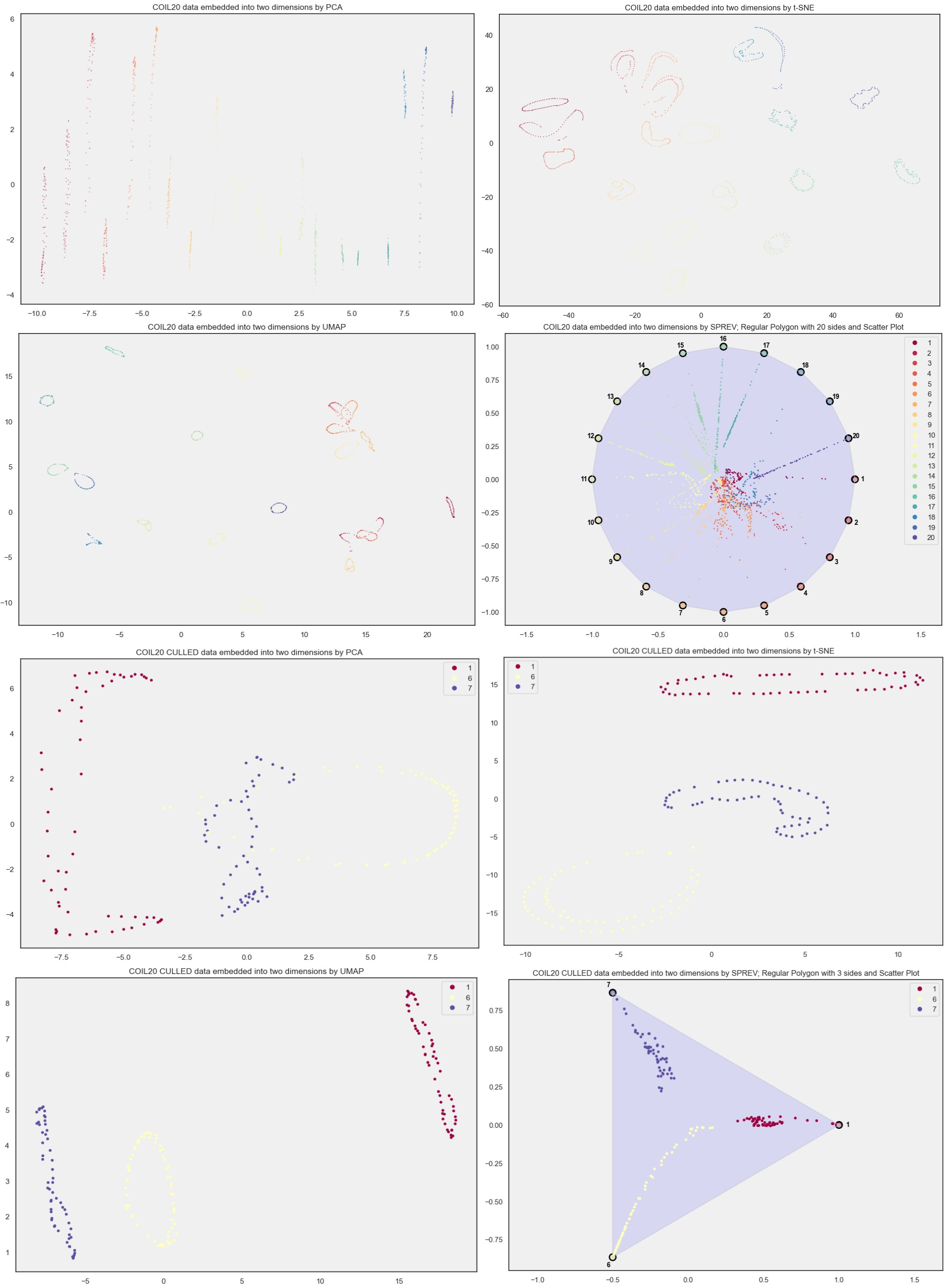}
  \caption{VISUALISATIONS of COIL-20 and COIL-20 CULLED embeddings for each of PCA, t-SNE, UMAP and SPREV.}
  \Description{Visualisation of embedding spaces for COIL-20 and COIL-20 CULLED}
  \label{COIL20-VIS}
\end{figure}

  \clearpage 
} 

\afterpage{
  \clearpage   

\begin{figure}[htb]
  \centering
  \includegraphics[width=0.93\textwidth,height=\textheight,keepaspectratio]{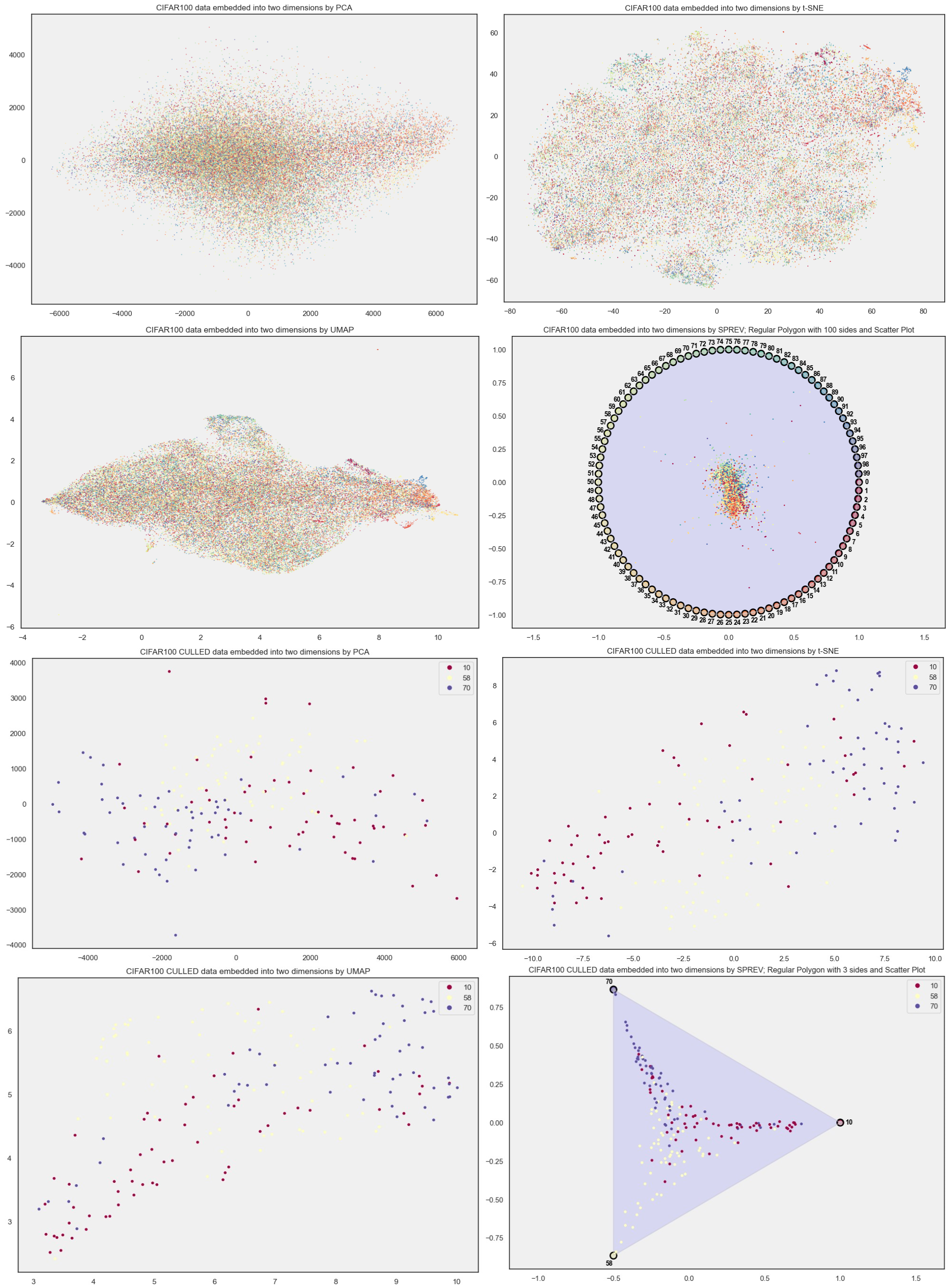}
  \caption{VISUALISATIONS of CIFAR-100 and CIFAR-100 CULLED embeddings for each of PCA, t-SNE, UMAP and SPREV.}
  \Description{Visualisation of embedding spaces for CIFAR-100 and CIFAR-100 CULLED}
  \label{CIFAR100-VIS}
\end{figure}

  \clearpage 
} 

\subsection{Limitations of Qualitative Comparison of Multiple Algorithms}
While the above section provides a starting point and provides insight into some aspects of the data  that were observed, it is imperative to acknowledge its limitations. Qualitative evaluation can be subjective and susceptible to individual biases. Developing reliability measures and incorporating more users observations and discussions could enhance the objectivity of the evaluation process. Furthermore, investigating these interactions and adapting the ideas from them would prove crucial for proper understanding.

\subsection{Quantiative Benchmarking of Embedding Space of the Algorithms}

\subsubsection{Separability benchmark via Linear SVM Classifier}
Our methodological choice to train a linear Support Vector Machine (SVM) classifier using the embedding spaces derived from the different dimensionality reduction techniques serves a purposeful role in the quantitative benchmarking via the assessment of class separation. The accuracy of a trained linear SVM model is informed by its inherent capability to discern linearly separable structures and patterns in feature spaces, aligning effectively with the characteristic required for clearer visualisation purposes using when using the various techniques. Our aim is to systematically scrutinize and juxtapose the performance of the chosen dimension reduction techniques, particularly in their efficacy to yield embedding spaces that foster effective separation of disparate classes. This emphasis on evaluating class separation is motivated by its pivotal significance in visualisation applications where discerning meaningful distinctions among categorical entities is of paramount importance. This discriminative prowess of the trained linear SVM serves as a judicious approach for  quantitatively gauging and assessing the extent to which embedding spaces, underpin robust and interpretable class separation for the purpose of rigorous quantitative benchmarking of the various techniques we would like to benchmark. Thus, a dimension reduction method that results in higher accuracy score in the linear SVM benchmark would show signs of greater class separability. The results of this test are observed in Table \ref{table:SVM-data-accuracy}.

\begin{table}[H]
\centering
\begin{tabular}{ll|rrrrr}
\toprule
\textbf{} & \textbf{Dataset} & \textbf{t-SNE} & \textbf{UMAP} & \textbf{PCA} & \textbf{SPREV} &\\
\cline{2-7}
\parbox[t]{2mm}{}
          & MNIST & 0.962{\tiny ($\pm$ 0.003)} & \textbf{0.965}{\tiny ($\pm$ 0.002)} & 0.445{\tiny ($\pm$ 0.007)} & 0.552{\tiny ($\pm$ 0.008)}   \\
\cline{2-7}
\parbox[t]{2mm}{}
          & MNIST CULLED & 0.851{\tiny ($\pm$ 0.097)} & 0.841{\tiny ($\pm$ 0.088)} & 0.754{\tiny ($\pm$ 0.117)} & \textbf{0.889}{\tiny ($\pm$ 0.050)} \\
\cline{2-7}
\parbox[t]{2mm}{}
          & Fashion-MNIST &  0.676{\tiny ($\pm$ 0.010)} & \textbf{0.691}{\tiny ($\pm$ 0.012)} & 0.520{\tiny ($\pm$ 0.017)} & 0.592{\tiny ($\pm$ 0.014)}   \\
\cline{2-7}
\parbox[t]{2mm}{}
          & Fashion-MNIST CULLED & 0.928{\tiny ($\pm$ 0.036)} & 0.928{\tiny ($\pm$ 0.043)} & 0.928{\tiny ($\pm$ 0.043)} & \textbf{0.933}{\tiny ($\pm$ 0.041)} \\
\cline{2-7}
\parbox[t]{2mm}{}
          & COIL-20 & 0.794{\tiny ($\pm$ 0.041)} & \textbf{0.834}{\tiny ($\pm$ 0.029)} & 0.593{\tiny ($\pm$ 0.051)} & 0.652{\tiny ($\pm$ 0.050)} \\
\cline{2-7}
\parbox[t]{2mm}{}
          & COIL-20 CULLED & \textbf{1.000}{\tiny ($\pm$ 0.000)} & \textbf{1.000}{\tiny ($\pm$ 0.000)} & 0.978{\tiny ($\pm$ 0.037)} & \textbf{1.000}{\tiny ($\pm$ 0.000)}\\
\cline{2-7}
\parbox[t]{2mm}{}
          & CIFAR-100 &  \textbf{0.047}{\tiny ($\pm$ 0.003)} & 0.044{\tiny ($\pm$ 0.003)} & 0.037{\tiny ($\pm$ 0.002)} & 0.038{\tiny ($\pm$ 0.002)}  \\
\cline{2-7}
\parbox[t]{2mm}{}
          & CIFAR-100 CULLED & 0.606{\tiny ($\pm$ 0.080)} & 0.572{\tiny ($\pm$ 0.132)} & 0.628{\tiny ($\pm$ 0.117)} & \textbf{0.655}{\tiny ($\pm$ 0.112)} \\

\bottomrule
\end{tabular}
\caption{SVM Classifier accuracy over the embedding spaces of all datasets. Average accuracy scores are given over 10-fold cross-validation for each of PCA, t-SNE and UMAP, SPREV.}\label{table:SVM-data-accuracy}
\end{table}

\subsubsection{Discussion of results from linear SVM benchmark}
During the SVM training, the omission of scaling in the embedding space data led to all algorithms, with the exception of PCA, concluding training within a reasonable timeframe (sub 30 minutes). In contrast, the PCA variant required an extensive duration of 490 hours after which it catastrophically crashed and remained incomplete. However, the implementation of scaling, which facilitated the normalization of features to a standard scale without introducing inherent bias, resulted in swift completion within a matter of minutes for the non-culled datasets and seconds for the culled datasets. This expedited convergence can be attributed to how we exploit the chain rule concept in differentiation during backpropagation, wherein the error function assumes a more spherical form after scaling. Consequently, gradient descent converges more rapidly due to more reduced and uniform curvature, necessitating fewer steps for convergence.

Post-scaling, the accuracy of the trained SVMs remained consistent for those that completed training. Notably, this observation is intriguing given that the SVM classifier lacks affine transformation invariance, implying that scaling features before training typically yields distinct outcomes. However, analysing our approach to feature scaling, achieved by normalizing each feature through mean centering and component-wise scaling to unit variance, we can see how undesirable artifacts were curbed.

\begin{figure}[htb]
  \centering
  \includegraphics[width=\linewidth]{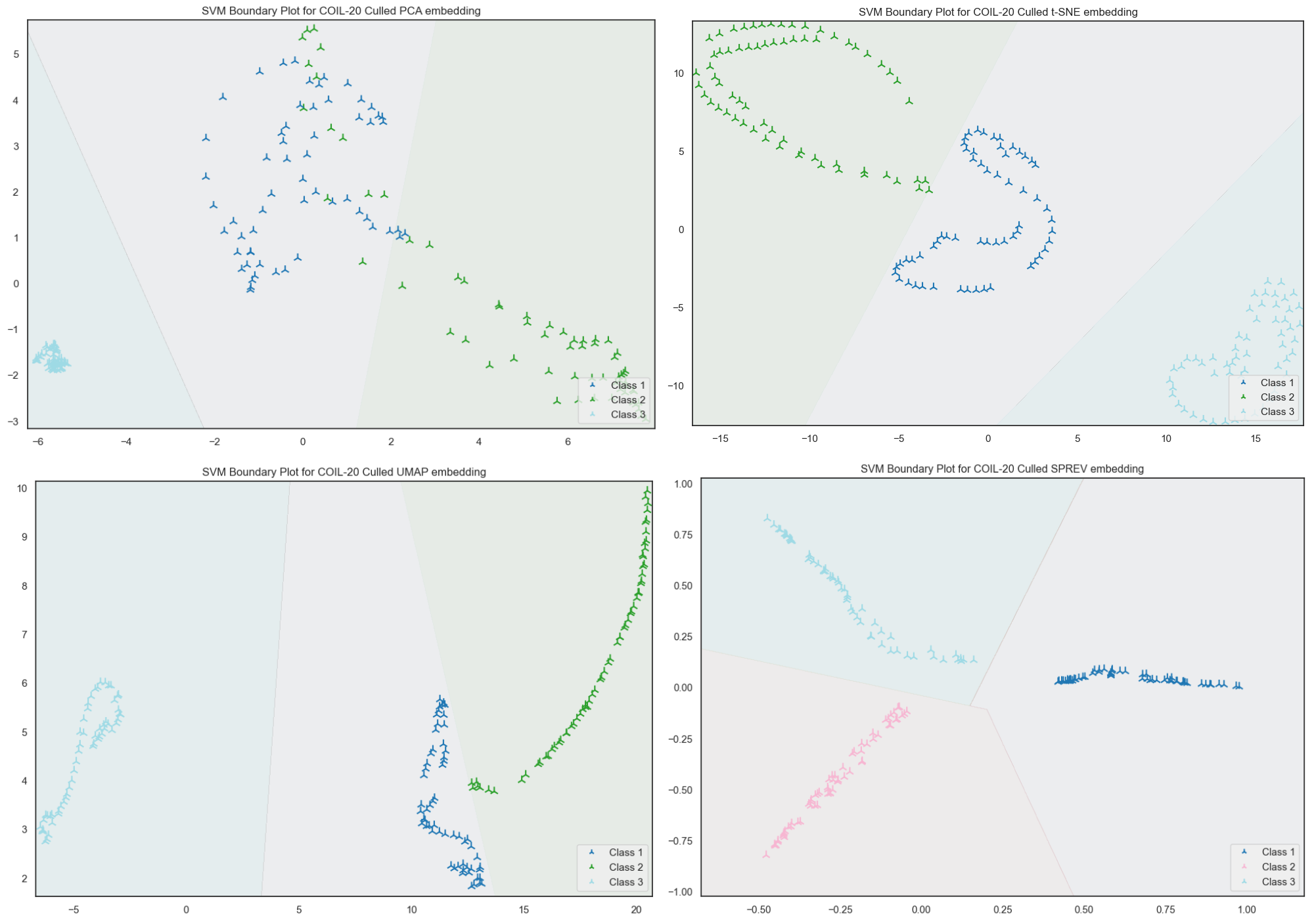}
  \caption{Graph showing boundary generated by the trained SVM on the embedding space.}
  \Description{Graph shows how the boundary is cleanly split for t-SNE, UMAP and SPREV. However for PCA this is not the case.}
  \label{SVM-graph}
\end{figure}

Upon scrutinizing the row for COIL-20 Dataset, an anomalous observation arises concerning the t-SNE, UMAP, SPREV, wherein a perfect accuracy is achieved. Such flawless accuracy typically indicates overfitting in certain contexts. when all three embeddings showcase distinct separation among the three classes, including their respective boundary regions. As depicted in Figure \ref{SVM-graph}, this peculiar behavior persists across different seed values for all employed dimension reduction methods when attempted without normalization. This recurrence suggests that the phenomenon is likely attributed to some other aspect of the dataset. Further discovery is required to elucidate the underlying reasons for this occurrence. However, as this would necessitate a more profound analysis of the dataset and an examination of the mechanics and internal workings of the SVM to identify specific pain points, delving into such intricacies is beyond the scope of this paper. For now, let us assume that the linear SVM might find it relatively easier to establish a separation that performs exceptionally well due to an undisclosed attribute present in the dataset post-culling. Leaving the exploration to understand the underlying reasons for this phenomenon to be deferred to the future work section.

When referencing with the qualitative benchmarks section, the accuracy attained by the classifier can be visually observed through the degree of separation achieved among individual classes in the embedding space. For instance, examining Figure \ref{CIFAR100-VIS} for CIFAR-100 reveals challenges in distinguishing points across all cases, manifesting as notably low accuracy scores for all employed techniques. In contrast, for COIL-20 CULLED, as illustrated in Figure \ref{COIL20-VIS}, clear class separation is evident, aligning with the high accuracy scores observed.

As an additional instance, considering MNIST, UMAP and t-SNE exhibit exceptional performance, corroborated by the well-defined separation of classes depicted in Figure \ref{MNIST-VIS}. Conversely, PCA yields a lower accuracy score, mirroring the densely packed point cloud where the separating lines are scarcely discernible. SPREV encounters a similar challenge towards the center, with points exhibiting less obvious separation, a characteristic duly reflected in the accuracy score – superior to PCA but notably inferior to t-SNE and UMAP.

\subsubsection{Local and Global Properties Preservation benchmark via kNN Classifier}
In this comprehensive investigation, we undertake a nuanced comparison of the embeddings derived from four distinct dimensionality reduction techniques:(\textit{PCA}), (\textit{UMAP}), (\textit{t-SNE}), and (\textit{SPREV}). Our primary focus revolves around evaluating the performance of a \(k\)-nearest neighbor (\textit{kNN}) classifier trained on the resulting embedding spaces across the chosen diverse array of datasets. The accuracy of the \textit{kNN} classifier serves as a robust quantitative measure, providing empirical insights into the extent to which each embedding preserves the crucial local structures inherent to the datasets.

The adaptability to manipulate the hyperparameter \(k\) during the training process affords us the opportunity to explore the spectrum of structure preservation, ranging from local to non-local and, ultimately, more global structures. This nuanced examination facilitates a deeper understanding of how each method interacts with the underlying data, capturing both intricate local details and overarching global patterns.

The datasets selected for this investigation encompass \textit{MNIST}, \textit{Fashion-MNIST}, \textit{COIL-20}, and \textit{CIFAR-100}, alongside their culled variants. The embedding evaluation is seperated into two distinct categories:

\begin{enumerate}
  \item \textbf{Culled Datasets:} This subset of datasets undergoes evaluation with a discerningly smaller range of \(k\) values, determined in accordance with the maximum \(k\) that is less than the size of all culled datasets, each comprising a standardized set of 180 samples. The application of exponential increments in \(k\) further refines our scrutiny of the structural transformations occurring within the embeddings.
  
  \item \textbf{Original Datasets:} These datasets necessitate a broader range of \(k\) values. The assignment of \(k\) values is intricately calibrated based on the sample sizes intrinsic to each dataset. For instance, \textit{MNIST} and \textit{Fashion-MNIST}, boasting 60,000 samples each, warrant a notably greater maximum \(k\). In contrast, \textit{COIL-20}, encompassing 1,440 samples, necessitates a smaller \(k\) value, albeit significantly larger than those applied to the culled datasets. \textit{CIFAR-100}, with its expansive 50,000 samples, justifies a large \(k\) value, albeit smaller than that of MNIST and Fashion-MNIST, to consistently explore the preservation of both local and global properties within the embeddings of each dataset.
\end{enumerate}

To ensure methodological rigor and consistency across all datasets, we employ a stratified 10-fold cross-validation framework. This approach yields a comprehensive set of 10 accuracy scores for each embedding, allowing for a detailed examination of the performance dynamics across diverse datasets. This rigorous evaluation methodology contributes to a robust understanding of the efficacy of different dimensionality reduction techniques in preserving the inherent structure of diverse datasets.

\begin{table}[H]
\centering
\begin{tabular}{ll|rrrrr}
\toprule
\textbf{} & \textbf{k} & \textbf{t-SNE} & \textbf{UMAP} & \textbf{PCA} & \textbf{SPREV} & \\
\cline{2-7}
\parbox[t]{2mm}{\multirow{6}{*}{\rotatebox[origin=c]{90}{MNIST}}}
&1500 & \textbf{0.963}{\tiny ($\pm$ 0.003)} & \textbf{0.963}{\tiny ($\pm$ 0.003)} & 0.464{\tiny ($\pm$ 0.006)} & 0.565{\tiny ($\pm$ 0.004)} \\
&3000 & 0.960{\tiny ($\pm$ 0.003)} & \textbf{0.961}{\tiny ($\pm$ 0.003)} & 0.456{\tiny ($\pm$ 0.006)} & 0.556{\tiny ($\pm$ 0.004)} \\
&5000 & 0.950{\tiny ($\pm$ 0.004)} & \textbf{0.953}{\tiny ($\pm$ 0.004)} & 0.450{\tiny ($\pm$ 0.007)} & 0.543{\tiny ($\pm$ 0.004)} \\
&15000 & \textbf{0.858}{\tiny ($\pm$ 0.006)} & 0.667{\tiny ($\pm$ 0.039)} & 0.430{\tiny ($\pm$ 0.006)} & 0.475{\tiny ($\pm$ 0.006)} \\
&30000 & 0.369{\tiny ($\pm$ 0.033)} & 0.317{\tiny ($\pm$ 0.006)} & 0.389{\tiny ($\pm$ 0.005)} & \textbf{0.423}{\tiny ($\pm$ 0.006)} \\
&50000 & 0.124{\tiny ($\pm$ 0.004)} & 0.189{\tiny ($\pm$ 0.005)} & 0.157{\tiny ($\pm$ 0.006)} & \textbf{0.264}{\tiny ($\pm$ 0.007)} \\

\cline{2-7}
\parbox[t]{2mm}{\multirow{6}{*}{\rotatebox[origin=c]{90}{Fashion-MNIST}}}
&1500 & \textbf{0.744}{\tiny ($\pm$ 0.006)} & 0.741{\tiny ($\pm$ 0.006)} & 0.547{\tiny ($\pm$ 0.006)} & 0.627{\tiny ($\pm$ 0.005)} \\
&3000 & 0.729{\tiny ($\pm$ 0.007)} & \textbf{0.736}{\tiny ($\pm$ 0.006)} & 0.529{\tiny ($\pm$ 0.007)} & 0.600{\tiny ($\pm$ 0.005)} \\
&5000 & \textbf{0.711}{\tiny ($\pm$ 0.007)} & 0.705{\tiny ($\pm$ 0.008)} & 0.507{\tiny ($\pm$ 0.008)} & 0.582{\tiny ($\pm$ 0.005)} \\
&15000 & \textbf{0.578}{\tiny ($\pm$ 0.006)} & 0.547{\tiny ($\pm$ 0.006)} & 0.431{\tiny ($\pm$ 0.007)} & 0.525{\tiny ($\pm$ 0.006)} \\
&30000 & 0.397{\tiny ($\pm$ 0.021)} & 0.406{\tiny ($\pm$ 0.031)} & 0.318{\tiny ($\pm$ 0.013)} & \textbf{0.435}{\tiny ($\pm$ 0.017)} \\
&50000 & 0.193{\tiny ($\pm$ 0.031)} & 0.189{\tiny ($\pm$ 0.035)} & 0.204{\tiny ($\pm$ 0.022)} & \textbf{0.286}{\tiny ($\pm$ 0.032)} \\

\cline{2-7}
\parbox[t]{2mm}{\multirow{6}{*}{\rotatebox[origin=c]{90}{COIL-20}}}
&30 & \textbf{0.878}{\tiny ($\pm$ 0.026)} & 0.841{\tiny ($\pm$ 0.020)} & 0.665{\tiny ($\pm$ 0.055)} & 0.725{\tiny ($\pm$ 0.031)} \\
&50 & 0.839{\tiny ($\pm$ 0.024)} & \textbf{0.842}{\tiny ($\pm$ 0.030)} & 0.627{\tiny ($\pm$ 0.041)} & 0.658{\tiny ($\pm$ 0.024)} \\
&150 & 0.589{\tiny ($\pm$ 0.030)} & \textbf{0.594}{\tiny ($\pm$ 0.065)} & 0.395{\tiny ($\pm$ 0.035)} & 0.501{\tiny ($\pm$ 0.045)} \\
&300 & 0.227{\tiny ($\pm$ 0.041)} & 0.280{\tiny ($\pm$ 0.077)} & 0.215{\tiny ($\pm$ 0.031)} & \textbf{0.311}{\tiny ($\pm$ 0.062)} \\
&600 & 0.083{\tiny ($\pm$ 0.027)} & 0.074{\tiny ($\pm$ 0.035)} & 0.108{\tiny ($\pm$ 0.026)} & \textbf{0.135}{\tiny ($\pm$ 0.037)} \\
&1200 & 0.023{\tiny ($\pm$ 0.009)} & 0.022{\tiny ($\pm$ 0.009)} & 0.029{\tiny ($\pm$ 0.010)} & \textbf{0.039}{\tiny ($\pm$ 0.022)} \\

\cline{2-7}

\parbox[t]{2mm}{\multirow{6}{*}{\rotatebox[origin=c]{90}{CIFAR-100}}}
&1250 & \textbf{0.075}{\tiny ($\pm$ 0.005)} & 0.066{\tiny ($\pm$ 0.005)} & 0.043{\tiny ($\pm$ 0.003)} & 0.043{\tiny ($\pm$ 0.002)} \\
&2500 &\textbf{} \textbf{0.064}{\tiny ($\pm$ 0.004)} & 0.055{\tiny ($\pm$ 0.004)} & 0.043{\tiny ($\pm$ 0.005)} & 0.042{\tiny ($\pm$ 0.002)} \\
&5000 & \textbf{0.049}{\tiny ($\pm$ 0.004)} & 0.048{\tiny ($\pm$ 0.003)} & 0.038{\tiny ($\pm$ 0.004)} & 0.037{\tiny ($\pm$ 0.003)} \\
&10000 & 0.037{\tiny ($\pm$ 0.004)} & \textbf{0.040}{\tiny ($\pm$ 0.003)} & 0.034{\tiny ($\pm$ 0.004)} & 0.033{\tiny ($\pm$ 0.002)} \\
&20000 & \textbf{0.033}{\tiny ($\pm$ 0.003)} & 0.031{\tiny ($\pm$ 0.003)} & 0.030{\tiny ($\pm$ 0.003)} & 0.029{\tiny ($\pm$ 0.003)} \\
&40000 & \textbf{0.016}{\tiny ($\pm$ 0.005)} & 0.010{\tiny ($\pm$ 0.004)} & 0.017{\tiny ($\pm$ 0.005)} & 0.012{\tiny ($\pm$ 0.003)} \\

\bottomrule
\end{tabular}
\caption{$k$NN Classifier accuracy for varying values of $k$ over the embedding spaces of MNIST, Fashion-MNIST, COIL-20
    and CIFAR-100 datasets. Average accuracy scores are given over a 10-fold cross-validation for each of PCA, t-SNE, UMAP and SPREV.}\label{table:original-data-accuracy}
\end{table}

\begin{table}[H]
\centering
\begin{tabular}{ll|rrrrr}
\toprule
\textbf{} & \textbf{k} & \textbf{t-SNE} & \textbf{UMAP} & \textbf{PCA} & \textbf{SPREV} & \\
\cline{2-7}
\parbox[t]{2mm}{\multirow{6}{*}{\rotatebox[origin=c]{90}{\footnotesize MNIST CULLED}}}
&3 & \textbf{0.933}{\tiny ($\pm$ 0.054)} & 0.906{\tiny ($\pm$ 0.090)} & 0.722{\tiny ($\pm$ 0.143)} & 0.853{\tiny ($\pm$ 0.058)} \\
&5 & \textbf{0.944}{\tiny ($\pm$ 0.050)} & 0.928{\tiny ($\pm$ 0.061)} & 0.761{\tiny ($\pm$ 0.108)} & 0.881{\tiny ($\pm$ 0.065)} \\
&15 & 0.883{\tiny ($\pm$ 0.076)} & 0.883{\tiny ($\pm$ 0.063)} & 0.767{\tiny ($\pm$ 0.089)} & \textbf{0.887}{\tiny ($\pm$ 0.072)} \\
&30 & \textbf{0.889}{\tiny ($\pm$ 0.066)} & 0.872{\tiny ($\pm$ 0.070)} & 0.756{\tiny ($\pm$ 0.094)} & 0.887{\tiny ($\pm$ 0.072)} \\
&60 & 0.861{\tiny ($\pm$ 0.057)} & 0.861{\tiny ($\pm$ 0.076)} & 0.750{\tiny ($\pm$ 0.100)} & \textbf{0.887}{\tiny ($\pm$ 0.077)} \\
&120 & 0.556{\tiny ($\pm$ 0.090)} & 0.617{\tiny ($\pm$ 0.094)} & 0.522{\tiny ($\pm$ 0.106)} & \textbf{0.715}{\tiny ($\pm$ 0.164)} \\
&150 & 0.383{\tiny ($\pm$ 0.118)} & \textbf{0.439}{\tiny ($\pm$ 0.098)} & 0.350{\tiny ($\pm$ 0.149)} & 0.401{\tiny ($\pm$ 0.047)} \\

\cline{2-7}
\parbox[t]{2mm}{\multirow{6}{*}{\rotatebox[origin=c]{90}{\footnotesize Fashion-MNIST CULLED}}}
&3 & \textbf{0.961}{\tiny ($\pm$ 0.056)} & 0.956{\tiny ($\pm$ 0.065)} & 0.922{\tiny ($\pm$ 0.079)} & 0.960{\tiny ($\pm$ 0.036)} \\
&5 & \textbf{0.967}{\tiny ($\pm$ 0.051)} & 0.956{\tiny ($\pm$ 0.054)} & 0.956{\tiny ($\pm$ 0.065)} & 0.960{\tiny ($\pm$ 0.036)} \\
&15 & \textbf{0.967}{\tiny ($\pm$ 0.051)} & 0.961{\tiny ($\pm$ 0.056)} & 0.944{\tiny ($\pm$ 0.082)} & 0.938{\tiny ($\pm$ 0.040)} \\
&30 & 0.956{\tiny ($\pm$ 0.048)} & \textbf{0.961}{\tiny ($\pm$ 0.056)} & 0.933{\tiny ($\pm$ 0.074)} & 0.932{\tiny ($\pm$ 0.049)} \\
&60 & 0.939{\tiny ($\pm$ 0.058)} & \textbf{0.956}{\tiny ($\pm$ 0.048)} & 0.906{\tiny ($\pm$ 0.086)} & 0.932{\tiny ($\pm$ 0.049)} \\
&120 & 0.456{\tiny ($\pm$ 0.124)} & 0.406{\tiny ($\pm$ 0.153)} & 0.811{\tiny ($\pm$ 0.090)} & \textbf{0.938}{\tiny ($\pm$ 0.040)} \\
&150 & 0.389{\tiny ($\pm$ 0.124)} & 0.406{\tiny ($\pm$ 0.153)} & 0.478{\tiny ($\pm$ 0.156)} & \textbf{0.621}{\tiny ($\pm$ 0.152)} \\

\cline{2-7}
\parbox[t]{2mm}{\multirow{6}{*}{\rotatebox[origin=c]{90}{\footnotesize COIL-20 CULLED}}}
&3 & \textbf{1.000}{\tiny ($\pm$ 0.000)} & \textbf{1.000}{\tiny ($\pm$ 0.000)} & 0.933{\tiny ($\pm$ 0.054)} & \textbf{1.000}{\tiny ($\pm$ 0.000)} \\
&5 & \textbf{1.000}{\tiny ($\pm$ 0.000)} & \textbf{1.000}{\tiny ($\pm$ 0.000)} & 0.922{\tiny ($\pm$ 0.057)} & \textbf{1.000}{\tiny ($\pm$ 0.000)} \\
&15 & \textbf{1.000}{\tiny ($\pm$ 0.000)} & \textbf{1.000}{\tiny ($\pm$ 0.000)} & 0.917{\tiny ($\pm$ 0.067)} & 0.994{\tiny ($\pm$ 0.017)} \\
&30 & 0.989{\tiny ($\pm$ 0.022)} & \textbf{1.000}{\tiny ($\pm$ 0.000)} & 0.883{\tiny ($\pm$ 0.058)} & 0.977{\tiny ($\pm$ 0.037)} \\
&60 & 0.989{\tiny ($\pm$ 0.022)} & \textbf{1.000}{\tiny ($\pm$ 0.000)} & 0.772{\tiny ($\pm$ 0.123)} & 0.944{\tiny ($\pm$ 0.066)} \\
&120 & 0.639{\tiny ($\pm$ 0.134)} & 0.683{\tiny ($\pm$ 0.096)} & 0.683{\tiny ($\pm$ 0.149)} & \textbf{0.933}{\tiny ($\pm$ 0.085)} \\
&150 & 0.611{\tiny ($\pm$ 0.099)} & 0.444{\tiny ($\pm$ 0.155)} & \textbf{0.583}{\tiny ($\pm$ 0.134)} & 0.573{\tiny ($\pm$ 0.207)} \\

\cline{2-7}

\parbox[t]{2mm}{\multirow{6}{*}{\rotatebox[origin=c]{90}{\footnotesize CIFAR-100 CULLED}}}
&3 & \textbf{0.683}{\tiny ($\pm$ 0.103)} & 0.628{\tiny ($\pm$ 0.153)} & 0.556{\tiny ($\pm$ 0.129)} & 0.667{\tiny ($\pm$ 0.107)} \\
&5 & \textbf{0.678}{\tiny ($\pm$ 0.108)} & 0.606{\tiny ($\pm$ 0.152)} & 0.583{\tiny ($\pm$ 0.100)} & 0.673{\tiny ($\pm$ 0.111)} \\
&15 & 0.678{\tiny ($\pm$ 0.116)} & 0.661{\tiny ($\pm$ 0.133)} & 0.589{\tiny ($\pm$ 0.094)} & \textbf{0.723}{\tiny ($\pm$ 0.100)} \\
&30 & 0.678{\tiny ($\pm$ 0.136)} & 0.689{\tiny ($\pm$ 0.114)} & 0.606{\tiny ($\pm$ 0.104)} & \textbf{0.717}{\tiny ($\pm$ 0.102)} \\
&60 & 0.656{\tiny ($\pm$ 0.119)} & 0.667{\tiny ($\pm$ 0.127)} & 0.606{\tiny ($\pm$ 0.107)} & \textbf{0.689}{\tiny ($\pm$ 0.108)} \\
&120 & 0.361{\tiny ($\pm$ 0.103)} & 0.367{\tiny ($\pm$ 0.100)} & 0.389{\tiny ($\pm$ 0.136)} & \textbf{0.549}{\tiny ($\pm$ 0.195)} \\
&150 & 0.367{\tiny ($\pm$ 0.112)} & 0.361{\tiny ($\pm$ 0.103)} & 0.361{\tiny ($\pm$ 0.103)} & \textbf{0.391}{\tiny ($\pm$ 0.149)} \\

\bottomrule
\end{tabular}
\caption{$k$NN Classifier accuracy for varying values of $k$ over the embedding spaces of MNIST, Fashion-MNIST, COIL-20
    and CIFAR-100 CULLED datasets. Average accuracy scores are given over a 10-fold cross-validation for each of PCA, t-SNE, UMAP and SPREV.}\label{table:culled-data-accuracy}
\end{table}
\subsubsection{Discussion of results from kNN benchmark.}

In Table \ref{table:original-data-accuracy}, we present the average cross-validation accuracy for the MNIST, Fashion-MNIST, COIl-20, and CIFAR-100 datasets. A discernible trend is observed wherein UMAP and t-SNE exhibit relatively superior performance, particularly at lower \(k\)-values. This aligns with expectations, given the inherent approach of t-SNE and UMAP in creating connectivity through pairwise similarities between data points in the original graphical representation before projecting it to a lower dimension. This indicates their proficiency in preserving local structures compared to PCA and SPREV. Notably, SPREV, in many cases, closely trails behind UMAP and t-SNE while outperforming PCA, suggesting a nuanced balance in capturing local structures, although not as adeptly as t-SNE and UMAP.

As the \(k\)-value increases to a more non-local level, a distinctive pattern emerges where SPREV surpasses all algorithms, demonstrating superior control in maintaining non-local structures after dimension reduction compared to its counterparts. This suggests a successful achievement of the objective to strike a balance between preserving local and global structures post-dimension reduction. In the non-local scenarios, PCA exhibits relatively favorable performance compared to t-SNE and UMAP, consistent with expectations stemming from PCA's reliance on a global structure for separation and dimension reduction.

Notably, for the CIFAR-100 dataset, there is a significant performance decline across all algorithms. This is likely attributed to the multitude of classes compressed into a confined setting, posing a formidable challenge for preserving both non-local and local structures. The visual inspection of data points in Figure \ref{CIFAR100-VIS} confirms their dense clustering in all instances, making it arduous to extract meaningful structure. Additionally, a general observation is that as class size increases, non-local structure preservation diminishes. This is evident in the lower performance at higher \(k\)-values for COIL-20 and CIFAR-100 in comparison to MNIST and Fashion-MNIST.

In Table \ref{table:culled-data-accuracy}, we present the average cross-validation accuracy for the culled variants of the MNIST, Fashion-MNIST, COIL-20, and CIFAR-100 datasets. An interesting trend emerges in which, at medium to higher k values, SPREV generally outperforms UMAP and t-SNE, except for the case of MNIST CULLED, where UMAP exhibits a slight advantage, and for  COIL-20, where PCA holds a similar edge. However, these differences are marginal. The same trends persist at smaller k values, mirroring the patterns observed in the non-culled variants, where t-SNE and UMAP excel at the local scale.

An intriguing observation pertains to the COIL-20 CULLED variant, displaying numerous perfect scores at low and medium k values. This phenomenon is likely a characteristic of the culled dataset, as evident in Figure \ref{COIL20-VIS}, where well-separated clusters exist. The occurrence of perfect accuracy scores, even with 10-fold cross-validation, is higher due to the dataset's inherent characteristics, influencing kNN classifier behavior as a result.

It's essential to highlight the impact of dataset culling, particularly for CIFAR-100, where the reduced class size significantly enhances the performance of all algorithms. Visualizations in Figure \ref{CIFAR100-VIS} align with these results, showing more discernible local and non-local trends in the culled variants, contrasting the almost noise like visualizations in the unculled counterparts, attributed to the sheer abundance of samples and classes. The improvement in results hence, corroborates the visual observations. Comparing the higher performance in Fashion-MNIST to MNIST, the better clustering in culled embeddings versus non-culled embeddings in Figure \ref{FMNIST-VIS} and Figure \ref{MNIST-VIS} reinforces our analysis and provides a quantitative perspective on the observed visual trends.

\subsection{Comparison of Computational Performance of Plotting Algorithms}

Benchmarking against real-world datasets was conducted on an ArchLinux Instance equipped with a 3.80 GHz 24-Core, 48-Thread AMD Ryzen Threadripper 3960X processor and 128 GiB of quad-channel DDR4-3200 RAM. The implementations of PCA, t-SNE, and UMAP utilized default non-multicore or multithreaded variants available in scikit-learn, chosen for the sake of reproducibility on simpler systems. While it was feasible to develop a variant of the SPREV algorithm leveraging GPU parallelism, a deliberate decision was made during testing to adhere to a straightforward CPU-bound approach. This choice aimed to maintain fairness and reproducibility, facilitating ease of use on more economical systems lacking dedicated GPUs.

\begin{table}[H]
\centering
\begin{tabular}{ll|rrrrr}
\toprule
\textbf{} & \textbf{Dataset} & \textbf{t-SNE} & \textbf{UMAP} & \textbf{PCA} & \textbf{SPREV} &\\
\cline{2-7}
\parbox[t]{2mm}{}
          & MNIST & 3m 2.2s & 1m 5.4s & 3.9s & \textbf{3.0s}\\
\cline{2-7}
\parbox[t]{2mm}{}
          & MNIST CULLED & 0.7s & 3.1s & \textbf{0.2s} & \textbf{0.2s} \\
\cline{2-7}
\parbox[t]{2mm}{}
          & Fashion-MNIST & 3m 14.2s & 46.4s & 3.8s & \textbf{2.7s}   \\
\cline{2-7}
\parbox[t]{2mm}{}
          & Fashion-MNIST CULLED & 0.7s & 1.1s & 0.5s & \textbf{0.2s} \\
\cline{2-7}
\parbox[t]{2mm}{}
          & COIL-20 & 3.5s & 3.8s & \textbf{0.2s} & 0.3s \\
\cline{2-7}
\parbox[t]{2mm}{}
          & COIL-20 CULLED & 0.7s & 1.0s & \textbf{0.2s} & \textbf{0.2s}  \\
\cline{2-7}
\parbox[t]{2mm}{}
          & CIFAR-100 & 4m 25.2s & 1m 4.6s & \textbf{4.7s} & 29.7s \\
\cline{2-7}
\parbox[t]{2mm}{}
          & CIFAR-100 CULLED & 0.8s & 1.3s & \textbf{0.2s} & \textbf{0.2s} \\

\bottomrule
\end{tabular}
\caption{Execution time for plot of all datasets to render for each of PCA, t-SNE and UMAP, SPREV.}\label{table:execution-time}
\end{table}

Upon examining Table \ref{table:execution-time}, a discernible pattern emerges regarding the CPU execution times until plots observed in Figures \ref{MNIST-VIS}, \ref{FMNIST-VIS}, \ref{COIL20-VIS} and \ref{CIFAR100-VIS} are rendered for the non-culled datasets. Notably, t-SNE exhibits the lengthiest processing duration, succeeded by UMAP, while PCA and SPREV demonstrate comparatively faster performances. An exception to this trend is observed in the case of CIFAR-100 and COIL-20, where SPREV is outperformed by PCA, indicating that, for scenarios characterized by low class sizes, SPREV exhibits exceptional efficiency, while PCA retains superior performance in variants with higher class sizes.

Conversely, when analyzing the culled dataset variants, a distinct trend manifests wherein UMAP registers as the slowest, trailed by t-SNE, PCA, and SPREV, respectively. Despite the relatively rapid completion times for all methods, this observation underscores SPREV's noteworthy performance, particularly in scenarios characterized by small class sizes, high dimensions, and low sample sizes. Moreover, SPREV exhibits promise even in high sample size scenarios, albeit encountering challenges as class sizes increase. Although SPREV maintains faster speeds compared to t-SNE and UMAP in such cases, it falls short of PCA's performance. Consequently, class size emerges as a notable weakness in the current SPREV implementation, suggesting potential avenues for enhancing its efficacy in high class size variants to yield more favorable outcomes and insights.

Exploring ways to augment the implementation for improved results and observations in high class size scenarios could be a viable future direction for research.

\section{Future Work}
Several promising directions emerge, each offering opportunities for refinement and expansion.

\subsection{Multicore and Multithreaded Optimization:}
An avenue ripe for exploration entails the optimization of SPREV for contemporary CPUs through strategic utilization of multicore and multithreaded architectures. Crafting a variant that judiciously leverages available processing cores is anticipated to significantly enhance overall performance, particularly in scenarios with ample computational resources.

\subsection{Behavior Analysis on Specific Datasets:}
Conducting a focused analysis on SPREV's behavior, particularly on datasets like COIL20, offers a pertinent avenue for deeper understanding. The recurrent occurrence of perfect accuracies post-culling demands meticulous investigation, aiming to unveil inherent dataset characteristics or shed light on the algorithm's behavior under specific conditions. Such insights promise to enrich our understanding of SPREV's capabilities and limitations. Moreover, conducting further testing with a different array of datasets could potentially unveil additional intriguing insights.

\subsection{Large Class Size Variant:}
The exploration of the large class size variant of SPREV presents an intriguing avenue for future work. While SPREV has demonstrated commendable efficiency in scenarios characterized by small class sizes, high dimensions, and low sample sizes, addressing the challenges associated with larger class sizes emerges as a vital next step. Investigating ways to enhance SPREV's performance and scalability in datasets with a substantial number of classes could lead to valuable insights. This endeavor could involve optimizing algorithms, refining computational strategies, or exploring novel approaches to better handle the increased complexity associated with larger class sizes through shape analysis. By delving into the intricacies of SPREV's behavior in datasets with varied class sizes, researchers can contribute to advancing the method's versatility and applicability across a broader spectrum of real-world scenarios.

\subsection{Automatic Choice of Bounding Space and Visualization Space Variant:}
In this iteration, we constrain our approach to utilizing a hypersphere as the high dimensional representation space and a regular polygon as the two-dimensional visualisation space. It would be intriguing to devise a methodology capable of autonomously exploring the dataset, discerning contextual cues related to the topology and (co)homology. This exploration aims to unveil the inherent shape and genus of the high-dimensional data, facilitating its transformation into a fundamental group. Subsequently, this fundamental group can be mapped to a two-dimensional object that optimally represents the data, enhancing comprehension and enabling the extraction of meaningful insights.

\subsection{Unsupervised Dataset Variant:}
In expanding the utility of SPREV, a prospective avenue of exploration involves the development of a variant specifically tailored for unsupervised datasets. The adaptation of SPREV to unsupervised scenarios necessitates a fundamental reevaluation of the algorithm's core components to accommodate the absence of labeled data. Addressing this unique challenge requires the formulation of novel methodologies for feature extraction, clustering, and pattern recognition that align with the inherent characteristics of unsupervised datasets. This avenue of research aligns with the overarching goal of augmenting SPREV ensuring its versatility in addressing real-world challenges.

\section{Conclusions}

We have devised a versatile dimension reduction technique tailored for labeled datasets, firmly rooted in robust mathematical principles fine-tuned through algorithmic optimisations for performant use in discrete computational settings. The algorithm implementing this technique exhibits notable computational efficiency compared to UMAP and t-SNE, while concurrently ensuring global structure preservation comparable to PCA. This capability enables the generation of high-quality embeddings for datasets characterized by small class sizes, high dimensions, and low sample sizes—an underexplored yet rapidly burgeoning realm within the domain of data science. The empirical assessment of SPREV's performance across diverse datasets underscores the algorithm's robustness, positioning it as a potentially transformative tool with significant implications for various scientific computing applications.

It is additionally furnished with a robust framework for visualizing the embedding space. This feature enables its application in comprehending relationships through observational analysis of the data, facilitating users such as data analysts, fast turnaround R\&D teams, and business intelligence personnel to employ it for promptly instigating discovery and effectively communicating findings with stakeholders.

However, SPREV exhibits certain limitations and weakness that we should note. Although the method aims to reduce dimensionality, the initial data needs to be in a sufficiently high-dimensional space for the convex hull and bounding sphere construction to be meaningful. In low-dimensional inital settings, these elements might not provide valuable insights into the data structure. While these elements contribute to preserving global structure, their effectiveness depends on the underlying data distribution. If the data exhibits significant outliers or complex shapes that deviate significantly from convexity, the generated sphere might not accurately represent the overall distribution, potentially leading to distortions in the visualization. The selection of an appropriate distance metric significantly influences the method's performance, as different metrics capture distinct aspects of similarity between data points. The choice of an inappropriate metric may lead to misleading or inaccurate representations of the data in the lower-dimensional space.

\begin{acks}
I would like to express my sincere gratitude to Ong Jun Wen\orcidlink{https://orcid.org/0009-0007-0870-604X} for his invaluable contribution and collaborative efforts in co-developing the foundational framework for the rigorous proof of orthogonality within high-dimensional spaces, as detailed in section \ref{section:Near-Orthogonality}. His dedication, expertise, and insightful perspectives significantly enriched the development and refinement of our previous research on high dimensional medical imagery denoising which laid the foundation and played a pivotal role in enabling this research endeavour of mine. This collaborative endeavor has not only enhanced the quality and depth of the work here but has also laid the core groundwork for the motivations behind the proofs presented in our study.
\end{acks}

\bibliographystyle{ACM-Reference-Format}
\bibliography{sample-base}


\begin{thebibliography}{69}


\ifx \showCODEN    \undefined \def \showCODEN     #1{\unskip}     \fi
\ifx \showDOI      \undefined \def \showDOI       #1{#1}\fi
\ifx \showISBNx    \undefined \def \showISBNx     #1{\unskip}     \fi
\ifx \showISBNxiii \undefined \def \showISBNxiii  #1{\unskip}     \fi
\ifx \showISSN     \undefined \def \showISSN      #1{\unskip}     \fi
\ifx \showLCCN     \undefined \def \showLCCN      #1{\unskip}     \fi
\ifx \shownote     \undefined \def \shownote      #1{#1}          \fi
\ifx \showarticletitle \undefined \def \showarticletitle #1{#1}   \fi
\ifx \showURL      \undefined \def \showURL       {\relax}        \fi
\providecommand\bibfield[2]{#2}
\providecommand\bibinfo[2]{#2}
\providecommand\natexlab[1]{#1}
\providecommand\showeprint[2][]{arXiv:#2}

\bibitem[Ahn et~al\mbox{.}(2007)]%
        {ahn2007high}
\bibfield{author}{\bibinfo{person}{Jeongyoun Ahn}, \bibinfo{person}{JS Marron}, \bibinfo{person}{Keith~M Muller}, {and} \bibinfo{person}{Yueh-Yun Chi}.} \bibinfo{year}{2007}\natexlab{}.
\newblock \showarticletitle{The high-dimension, low-sample-size geometric representation holds under mild conditions}.
\newblock \bibinfo{journal}{\emph{Biometrika}} \bibinfo{volume}{94}, \bibinfo{number}{3} (\bibinfo{year}{2007}), \bibinfo{pages}{760--766}.
\newblock


\bibitem[Ali and Ar{\i}t{\"u}rk(2014)]%
        {ali2014dynamic}
\bibfield{author}{\bibinfo{person}{{\"O}zden~G{\"u}r Ali} {and} \bibinfo{person}{Umut Ar{\i}t{\"u}rk}.} \bibinfo{year}{2014}\natexlab{}.
\newblock \showarticletitle{Dynamic churn prediction framework with more effective use of rare event data: The case of private banking}.
\newblock \bibinfo{journal}{\emph{Expert Systems with Applications}} \bibinfo{volume}{41}, \bibinfo{number}{17} (\bibinfo{year}{2014}), \bibinfo{pages}{7889--7903}.
\newblock


\bibitem[Asadoorian and Kantarelis(2005)]%
        {asadoorian2005essentials}
\bibfield{author}{\bibinfo{person}{Malcolm~O Asadoorian} {and} \bibinfo{person}{Demetrius Kantarelis}.} \bibinfo{year}{2005}\natexlab{}.
\newblock \bibinfo{booktitle}{\emph{Essentials of inferential statistics}}.
\newblock \bibinfo{publisher}{University Press of America}.
\newblock


\bibitem[Axler(2015)]%
        {axler2015linear}
\bibfield{author}{\bibinfo{person}{Sheldon Axler}.} \bibinfo{year}{2015}\natexlab{}.
\newblock \bibinfo{booktitle}{\emph{Linear algebra done right}}.
\newblock \bibinfo{publisher}{Springer}.
\newblock


\bibitem[Baldominos et~al\mbox{.}(2019)]%
        {baldominos2019survey}
\bibfield{author}{\bibinfo{person}{Alejandro Baldominos}, \bibinfo{person}{Yago Saez}, {and} \bibinfo{person}{Pedro Isasi}.} \bibinfo{year}{2019}\natexlab{}.
\newblock \showarticletitle{A survey of handwritten character recognition with mnist and emnist}.
\newblock \bibinfo{journal}{\emph{Applied Sciences}} \bibinfo{volume}{9}, \bibinfo{number}{15} (\bibinfo{year}{2019}), \bibinfo{pages}{3169}.
\newblock


\bibitem[Banerjee et~al\mbox{.}(2009)]%
        {banerjee2009hypothesis}
\bibfield{author}{\bibinfo{person}{Amitav Banerjee}, \bibinfo{person}{UB Chitnis}, \bibinfo{person}{SL Jadhav}, \bibinfo{person}{JS Bhawalkar}, {and} \bibinfo{person}{S Chaudhury}.} \bibinfo{year}{2009}\natexlab{}.
\newblock \showarticletitle{Hypothesis testing, type I and type II errors}.
\newblock \bibinfo{journal}{\emph{Industrial psychiatry journal}} \bibinfo{volume}{18}, \bibinfo{number}{2} (\bibinfo{year}{2009}), \bibinfo{pages}{127}.
\newblock


\bibitem[Button et~al\mbox{.}(2013)]%
        {button2013power}
\bibfield{author}{\bibinfo{person}{Katherine~S Button}, \bibinfo{person}{John~PA Ioannidis}, \bibinfo{person}{Claire Mokrysz}, \bibinfo{person}{Brian~A Nosek}, \bibinfo{person}{Jonathan Flint}, \bibinfo{person}{Emma~SJ Robinson}, {and} \bibinfo{person}{Marcus~R Munaf{\`o}}.} \bibinfo{year}{2013}\natexlab{}.
\newblock \showarticletitle{Power failure: why small sample size undermines the reliability of neuroscience}.
\newblock \bibinfo{journal}{\emph{Nature reviews neuroscience}} \bibinfo{volume}{14}, \bibinfo{number}{5} (\bibinfo{year}{2013}), \bibinfo{pages}{365--376}.
\newblock


\bibitem[Chen et~al\mbox{.}(2007)]%
        {chen2007handbook}
\bibfield{author}{\bibinfo{person}{Chun-houh Chen}, \bibinfo{person}{Wolfgang H{\"a}rdle}, {and} \bibinfo{person}{Antony Unwin}.} \bibinfo{year}{2007}\natexlab{}.
\newblock \bibinfo{booktitle}{\emph{Handbook of Data Visualization}}.
\newblock \bibinfo{publisher}{Springer}.
\newblock


\bibitem[Choi et~al\mbox{.}(2021)]%
        {choi2021boosting}
\bibfield{author}{\bibinfo{person}{Stephen~J Choi}, \bibinfo{person}{Xinyue Cui}, {and} \bibinfo{person}{Jingran Zhao}.} \bibinfo{year}{2021}\natexlab{}.
\newblock \showarticletitle{Boosting over Deep Learning for Earnings Stephen J. Choi, Xinyue Cui, 2 Jingran Zhao3}.
\newblock  (\bibinfo{year}{2021}).
\newblock


\bibitem[Cohen(1992)]%
        {cohen1992statistical}
\bibfield{author}{\bibinfo{person}{Jacob Cohen}.} \bibinfo{year}{1992}\natexlab{}.
\newblock \showarticletitle{Statistical power analysis}.
\newblock \bibinfo{journal}{\emph{Current directions in psychological science}} \bibinfo{volume}{1}, \bibinfo{number}{3} (\bibinfo{year}{1992}), \bibinfo{pages}{98--101}.
\newblock


\bibitem[Cohen(2013)]%
        {cohen2013statistical}
\bibfield{author}{\bibinfo{person}{Jacob Cohen}.} \bibinfo{year}{2013}\natexlab{}.
\newblock \bibinfo{booktitle}{\emph{Statistical power analysis for the behavioral sciences}}.
\newblock \bibinfo{publisher}{Academic press}.
\newblock


\bibitem[Crow et~al\mbox{.}(2006)]%
        {crow2006research}
\bibfield{author}{\bibinfo{person}{Graham Crow}, \bibinfo{person}{Rose Wiles}, \bibinfo{person}{Sue Heath}, {and} \bibinfo{person}{Vikki Charles}.} \bibinfo{year}{2006}\natexlab{}.
\newblock \showarticletitle{Research ethics and data quality: The implications of informed consent}.
\newblock \bibinfo{journal}{\emph{International journal of social research methodology}} \bibinfo{volume}{9}, \bibinfo{number}{2} (\bibinfo{year}{2006}), \bibinfo{pages}{83--95}.
\newblock


\bibitem[dohmatob (https://math.stackexchange.com/users/168758/dohmatob)({[n.\,d.]})]%
        {1807022}
\bibfield{author}{\bibinfo{person}{dohmatob (https://math.stackexchange.com/users/168758/dohmatob)}.} \bibinfo{year}{[n.\,d.]}\natexlab{}.
\newblock \bibinfo{title}{The standard n-simplex is compact set}.
\newblock \bibinfo{howpublished}{Mathematics Stack Exchange}.
\newblock
\showeprint{https://math.stackexchange.com/q/1807022}
\urldef\tempurl%
\url{https://math.stackexchange.com/q/1807022}
\showURL{%
\tempurl}
\newblock
\shownote{URL:https://math.stackexchange.com/q/1807022 (version: 2017-04-13)}.


\bibitem[Easterbrook et~al\mbox{.}(1991)]%
        {easterbrook1991publication}
\bibfield{author}{\bibinfo{person}{Phillipa~J Easterbrook}, \bibinfo{person}{Ramana Gopalan}, \bibinfo{person}{JA Berlin}, {and} \bibinfo{person}{David~R Matthews}.} \bibinfo{year}{1991}\natexlab{}.
\newblock \showarticletitle{Publication bias in clinical research}.
\newblock \bibinfo{journal}{\emph{The Lancet}} \bibinfo{volume}{337}, \bibinfo{number}{8746} (\bibinfo{year}{1991}), \bibinfo{pages}{867--872}.
\newblock


\bibitem[Eisen and Brown(1999)]%
        {eisen199912}
\bibfield{author}{\bibinfo{person}{Michael~B Eisen} {and} \bibinfo{person}{Patrick~O Brown}.} \bibinfo{year}{1999}\natexlab{}.
\newblock \showarticletitle{[12] DNA arrays for analysis of gene expression}.
\newblock In \bibinfo{booktitle}{\emph{Methods in enzymology}}. Vol.~\bibinfo{volume}{303}. \bibinfo{publisher}{Elsevier}, \bibinfo{pages}{179--205}.
\newblock


\bibitem[Fitzpatrick(2009)]%
        {fitzpatrick2009advanced}
\bibfield{author}{\bibinfo{person}{Patrick Fitzpatrick}.} \bibinfo{year}{2009}\natexlab{}.
\newblock \bibinfo{booktitle}{\emph{Advanced calculus}}. Vol.~\bibinfo{volume}{5}.
\newblock \bibinfo{publisher}{American Mathematical Soc.}
\newblock


\bibitem[Graham(1972)]%
        {graham1972efficient}
\bibfield{author}{\bibinfo{person}{Ronald~L. Graham}.} \bibinfo{year}{1972}\natexlab{}.
\newblock \showarticletitle{An efficient algorithm for determining the convex hull of a finite planar set}.
\newblock \bibinfo{journal}{\emph{Info. Proc. Lett.}}  \bibinfo{volume}{1} (\bibinfo{year}{1972}), \bibinfo{pages}{132--133}.
\newblock


\bibitem[Halmos(1974)]%
        {halmos1974finite}
\bibfield{author}{\bibinfo{person}{Paul~Richard Halmos}.} \bibinfo{year}{1974}\natexlab{}.
\newblock \bibinfo{booktitle}{\emph{Finite-dimensional vector spaces}}.
\newblock \bibinfo{publisher}{Springer-Verlag}.
\newblock


\bibitem[Harrington et~al\mbox{.}(2000)]%
        {harrington2000monitoring}
\bibfield{author}{\bibinfo{person}{Christina~A Harrington}, \bibinfo{person}{Carsten Rosenow}, {and} \bibinfo{person}{Jacques Retief}.} \bibinfo{year}{2000}\natexlab{}.
\newblock \showarticletitle{Monitoring gene expression using DNA microarrays}.
\newblock \bibinfo{journal}{\emph{Current opinion in Microbiology}} \bibinfo{volume}{3}, \bibinfo{number}{3} (\bibinfo{year}{2000}), \bibinfo{pages}{285--291}.
\newblock


\bibitem[Hertzog(2008)]%
        {hertzog2008considerations}
\bibfield{author}{\bibinfo{person}{Melody~A Hertzog}.} \bibinfo{year}{2008}\natexlab{}.
\newblock \showarticletitle{Considerations in determining sample size for pilot studies}.
\newblock \bibinfo{journal}{\emph{Research in nursing \& health}} \bibinfo{volume}{31}, \bibinfo{number}{2} (\bibinfo{year}{2008}), \bibinfo{pages}{180--191}.
\newblock


\bibitem[Hoeffding(1963)]%
        {doi:10.1080/01621459.1963.10500830}
\bibfield{author}{\bibinfo{person}{Wassily Hoeffding}.} \bibinfo{year}{1963}\natexlab{}.
\newblock \showarticletitle{Probability Inequalities for Sums of Bounded Random Variables}.
\newblock \bibinfo{journal}{\emph{J. Amer. Statist. Assoc.}} \bibinfo{volume}{58}, \bibinfo{number}{301} (\bibinfo{year}{1963}), \bibinfo{pages}{13--30}.
\newblock
\urldef\tempurl%
\url{https://doi.org/10.1080/01621459.1963.10500830}
\showDOI{\tempurl}
\showeprint{https://www.tandfonline.com/doi/pdf/10.1080/01621459.1963.10500830}


\bibitem[(https://mathoverflow.net/users/1409/mariano-sua1rez alvarez)({[n.\,d.]})]%
        {22608}
\bibfield{author}{\bibinfo{person}{Mariano Suarez-Alvarez (https://mathoverflow.net/users/1409/mariano-sua1rez alvarez)}.} \bibinfo{year}{[n.\,d.]}\natexlab{}.
\newblock \bibinfo{title}{convex hull of finite set is compact}.
\newblock \bibinfo{howpublished}{MathOverflow}.
\newblock
\showeprint{https://mathoverflow.net/q/22608}
\urldef\tempurl%
\url{https://mathoverflow.net/q/22608}
\showURL{%
\tempurl}
\newblock
\shownote{URL:https://mathoverflow.net/q/22608 (version: 2010-04-26)}.


\bibitem[(https://math.stackexchange.com/users/1172204/jay mody)({[n.\,d.]})]%
        {4682000}
\bibfield{author}{\bibinfo{person}{Jay~Mody (https://math.stackexchange.com/users/1172204/jay mody)}.} \bibinfo{year}{[n.\,d.]}\natexlab{}.
\newblock \bibinfo{title}{Can all vectors be described with a tuple?}
\newblock \bibinfo{howpublished}{Mathematics Stack Exchange}.
\newblock
\showeprint{https://math.stackexchange.com/q/4682000}
\urldef\tempurl%
\url{https://math.stackexchange.com/q/4682000}
\showURL{%
\tempurl}
\newblock
\shownote{URL:https://math.stackexchange.com/q/4682000 (version: 2023-04-19)}.


\bibitem[(https://math.stackexchange.com/users/131617/ssf)({[n.\,d.]})]%
        {1603122}
\bibfield{author}{\bibinfo{person}{SSF (https://math.stackexchange.com/users/131617/ssf)}.} \bibinfo{year}{[n.\,d.]}\natexlab{}.
\newblock \bibinfo{title}{Length of the main diagonal of an n-dimensional cube}.
\newblock \bibinfo{howpublished}{Mathematics Stack Exchange}.
\newblock
\showeprint{https://math.stackexchange.com/q/1603122}
\urldef\tempurl%
\url{https://math.stackexchange.com/q/1603122}
\showURL{%
\tempurl}
\newblock
\shownote{URL:https://math.stackexchange.com/q/1603122 (version: 2016-01-07)}.


\bibitem[(https://math.stackexchange.com/users/405819/adren)({[n.\,d.]})]%
        {2100328}
\bibfield{author}{\bibinfo{person}{Adren (https://math.stackexchange.com/users/405819/adren)}.} \bibinfo{year}{[n.\,d.]}\natexlab{}.
\newblock \bibinfo{title}{Smallest closed ball enclosing a compact set}.
\newblock \bibinfo{howpublished}{Mathematics Stack Exchange}.
\newblock
\showeprint{https://math.stackexchange.com/q/2100328}
\urldef\tempurl%
\url{https://math.stackexchange.com/q/2100328}
\showURL{%
\tempurl}
\newblock
\shownote{URL:https://math.stackexchange.com/q/2100328 (version: 2017-01-18)}.


\bibitem[Jackson and Landgrebe(2001)]%
        {jackson2001adaptive}
\bibfield{author}{\bibinfo{person}{Qiong Jackson} {and} \bibinfo{person}{David~A Landgrebe}.} \bibinfo{year}{2001}\natexlab{}.
\newblock \showarticletitle{An adaptive classifier design for high-dimensional data analysis with a limited training data set}.
\newblock \bibinfo{journal}{\emph{IEEE Transactions on Geoscience and Remote Sensing}} \bibinfo{volume}{39}, \bibinfo{number}{12} (\bibinfo{year}{2001}), \bibinfo{pages}{2664--2679}.
\newblock


\bibitem[Jamjoom(2021)]%
        {jamjoom2021use}
\bibfield{author}{\bibinfo{person}{Arwa~A Jamjoom}.} \bibinfo{year}{2021}\natexlab{}.
\newblock \showarticletitle{The use of knowledge extraction in predicting customer churn in B2B}.
\newblock \bibinfo{journal}{\emph{Journal of Big Data}} \bibinfo{volume}{8}, \bibinfo{number}{1} (\bibinfo{year}{2021}), \bibinfo{pages}{110}.
\newblock


\bibitem[Jarvis(1973)]%
        {jarvis1973identification}
\bibfield{author}{\bibinfo{person}{Ray~A Jarvis}.} \bibinfo{year}{1973}\natexlab{}.
\newblock \showarticletitle{On the identification of the convex hull of a finite set of points in the plane}.
\newblock \bibinfo{journal}{\emph{Information processing letters}} \bibinfo{volume}{2}, \bibinfo{number}{1} (\bibinfo{year}{1973}), \bibinfo{pages}{18--21}.
\newblock


\bibitem[Johnson(2012)]%
        {johnson2012notes}
\bibfield{author}{\bibinfo{person}{Steven~G Johnson}.} \bibinfo{year}{2012}\natexlab{}.
\newblock \showarticletitle{Notes on the equivalence of norms}.
\newblock \bibinfo{journal}{\emph{MIT Course 18.335}} (\bibinfo{year}{2012}), \bibinfo{pages}{1--2}.
\newblock


\bibitem[Kelly(2014)]%
        {kelly2014arrow}
\bibfield{author}{\bibinfo{person}{Jerry~S Kelly}.} \bibinfo{year}{2014}\natexlab{}.
\newblock \bibinfo{booktitle}{\emph{Arrow impossibility theorems}}.
\newblock \bibinfo{publisher}{Academic Press}.
\newblock


\bibitem[Klose et~al\mbox{.}(2020)]%
        {klose2020edm}
\bibfield{author}{\bibinfo{person}{Mark Klose}, \bibinfo{person}{Vasvi Desai}, \bibinfo{person}{Yang Song}, {and} \bibinfo{person}{Edward Gehringer}.} \bibinfo{year}{2020}\natexlab{}.
\newblock \showarticletitle{EDM and Privacy: Ethics and Legalities of Data Collection, Usage, and Storage.}
\newblock \bibinfo{journal}{\emph{International Educational Data Mining Society}} (\bibinfo{year}{2020}).
\newblock


\bibitem[Konietschke et~al\mbox{.}(2021)]%
        {konietschke2021small}
\bibfield{author}{\bibinfo{person}{Frank Konietschke}, \bibinfo{person}{Karima Schwab}, {and} \bibinfo{person}{Markus Pauly}.} \bibinfo{year}{2021}\natexlab{}.
\newblock \showarticletitle{Small sample sizes: A big data problem in high-dimensional data analysis}.
\newblock \bibinfo{journal}{\emph{Statistical Methods in Medical Research}} \bibinfo{volume}{30}, \bibinfo{number}{3} (\bibinfo{year}{2021}), \bibinfo{pages}{687--701}.
\newblock


\bibitem[Krizhevsky et~al\mbox{.}({[n.\,d.]})]%
        {CIFAR100cite}
\bibfield{author}{\bibinfo{person}{Alex Krizhevsky}, \bibinfo{person}{Vinod Nair}, {and} \bibinfo{person}{Geoffrey Hinton}.} \bibinfo{year}{[n.\,d.]}\natexlab{}.
\newblock \showarticletitle{CIFAR-100 (Canadian Institute for Advanced Research)}.
\newblock  (\bibinfo{year}{[n.\,d.]}).
\newblock
\urldef\tempurl%
\url{http://www.cs.toronto.edu/~kriz/cifar.html}
\showURL{%
\tempurl}


\bibitem[Krotov and Johnson(2023)]%
        {krotov2023big}
\bibfield{author}{\bibinfo{person}{Vlad Krotov} {and} \bibinfo{person}{Leigh Johnson}.} \bibinfo{year}{2023}\natexlab{}.
\newblock \showarticletitle{Big web data: Challenges related to data, technology, legality, and ethics}.
\newblock \bibinfo{journal}{\emph{Business Horizons}} \bibinfo{volume}{66}, \bibinfo{number}{4} (\bibinfo{year}{2023}), \bibinfo{pages}{481--491}.
\newblock


\bibitem[Landis and Koch(1977)]%
        {landis1977measurement}
\bibfield{author}{\bibinfo{person}{J~Richard Landis} {and} \bibinfo{person}{Gary~G Koch}.} \bibinfo{year}{1977}\natexlab{}.
\newblock \showarticletitle{The measurement of observer agreement for categorical data}.
\newblock \bibinfo{journal}{\emph{biometrics}} (\bibinfo{year}{1977}), \bibinfo{pages}{159--174}.
\newblock


\bibitem[LeCun~et al.({[n.\,d.]})]%
        {MNISTcite}
\bibfield{author}{\bibinfo{person}{1998a LeCun~et al.}} \bibinfo{year}{[n.\,d.]}\natexlab{}.
\newblock \showarticletitle{MNIST}.
\newblock  (\bibinfo{year}{[n.\,d.]}).
\newblock
\urldef\tempurl%
\url{http://yann.lecun.com/exdb/publis/index.html#lecun-98}
\showURL{%
\tempurl}


\bibitem[LEITHARDT(2021)]%
        {leithardt2021classifying}
\bibfield{author}{\bibinfo{person}{VALDERI LEITHARDT}.} \bibinfo{year}{2021}\natexlab{}.
\newblock \showarticletitle{Classifying garments from fashion-MNIST dataset through CNNs}.
\newblock \bibinfo{journal}{\emph{Advances in Science, Technology and Engineering Systems Journal}} \bibinfo{volume}{6}, \bibinfo{number}{1} (\bibinfo{year}{2021}), \bibinfo{pages}{989--994}.
\newblock


\bibitem[Li et~al\mbox{.}(2014)]%
        {6974480}
\bibfield{author}{\bibinfo{person}{Xinyang Li}, \bibinfo{person}{Wei Xu}, {and} \bibinfo{person}{Xuesong Tian}.} \bibinfo{year}{2014}\natexlab{}.
\newblock \showarticletitle{How to protect investors? A GA-based DWD approach for financial statement fraud detection}. In \bibinfo{booktitle}{\emph{2014 IEEE International Conference on Systems, Man, and Cybernetics (SMC)}}. \bibinfo{pages}{3548--3554}.
\newblock
\urldef\tempurl%
\url{https://doi.org/10.1109/SMC.2014.6974480}
\showDOI{\tempurl}


\bibitem[Mahmud et~al\mbox{.}(2019)]%
        {mahmud2019high}
\bibfield{author}{\bibinfo{person}{Mohammad~Sultan Mahmud}, \bibinfo{person}{Xianghua Fu}, \bibinfo{person}{Joshua~Zhexue Huang}, {and} \bibinfo{person}{Md~Abdul Masud}.} \bibinfo{year}{2019}\natexlab{}.
\newblock \showarticletitle{High-Dimensional Limited-Sample biomedical data classification using variational autoencoder}. In \bibinfo{booktitle}{\emph{Data Mining: 16th Australasian Conference, AusDM 2018, Bahrurst, NSW, Australia, November 28--30, 2018, Revised Selected Papers 16}}. Springer, \bibinfo{pages}{30--42}.
\newblock


\bibitem[Marron et~al\mbox{.}(2007)]%
        {marron2007distance}
\bibfield{author}{\bibinfo{person}{James~Stephen Marron}, \bibinfo{person}{Michael~J Todd}, {and} \bibinfo{person}{Jeongyoun Ahn}.} \bibinfo{year}{2007}\natexlab{}.
\newblock \showarticletitle{Distance-weighted discrimination}.
\newblock \bibinfo{journal}{\emph{J. Amer. Statist. Assoc.}} \bibinfo{volume}{102}, \bibinfo{number}{480} (\bibinfo{year}{2007}), \bibinfo{pages}{1267--1271}.
\newblock


\bibitem[martini (https://math.stackexchange.com/users/15379/martini)({[n.\,d.]})]%
        {1509256}
\bibfield{author}{\bibinfo{person}{martini (https://math.stackexchange.com/users/15379/martini)}.} \bibinfo{year}{[n.\,d.]}\natexlab{}.
\newblock \bibinfo{title}{Show that convex hull of a finite set is compact}.
\newblock \bibinfo{howpublished}{Mathematics Stack Exchange}.
\newblock
\showeprint{https://math.stackexchange.com/q/1509256}
\urldef\tempurl%
\url{https://math.stackexchange.com/q/1509256}
\showURL{%
\tempurl}
\newblock
\shownote{URL:https://math.stackexchange.com/q/1509256 (version: 2015-11-02)}.


\bibitem[Matas et~al\mbox{.}(2000)]%
        {matas2000object}
\bibfield{author}{\bibinfo{person}{Jiri Matas}, \bibinfo{person}{Jan Burianek}, {and} \bibinfo{person}{Josef Kittler}.} \bibinfo{year}{2000}\natexlab{}.
\newblock \showarticletitle{Object Recognition using the Invariant Pixel-Set Signature.}. In \bibinfo{booktitle}{\emph{BMVC}}, Vol.~\bibinfo{volume}{2}. \bibinfo{pages}{606--615}.
\newblock


\bibitem[McInnes et~al\mbox{.}(2018)]%
        {mcinnes2018umap}
\bibfield{author}{\bibinfo{person}{Leland McInnes}, \bibinfo{person}{John Healy}, {and} \bibinfo{person}{James Melville}.} \bibinfo{year}{2018}\natexlab{}.
\newblock \showarticletitle{UMAP: Uniform manifold approximation and projection for dimension reduction}.
\newblock \bibinfo{journal}{\emph{arXiv preprint arXiv:1802.05604}} (\bibinfo{year}{2018}).
\newblock


\bibitem[Meshkini et~al\mbox{.}(2020)]%
        {meshkini2020analysis}
\bibfield{author}{\bibinfo{person}{Khatereh Meshkini}, \bibinfo{person}{Jan Platos}, {and} \bibinfo{person}{Hassan Ghassemain}.} \bibinfo{year}{2020}\natexlab{}.
\newblock \showarticletitle{An analysis of convolutional neural network for fashion images classification (fashion-mnist)}. In \bibinfo{booktitle}{\emph{Proceedings of the Fourth International Scientific Conference “Intelligent Information Technologies for Industry”(IITI’19) 4}}. Springer, \bibinfo{pages}{85--95}.
\newblock


\bibitem[Nene et~al\mbox{.}(1996)]%
        {nene1996columbia}
\bibfield{author}{\bibinfo{person}{Sameer~A Nene}, \bibinfo{person}{Shree~K Nayar}, \bibinfo{person}{Hiroshi Murase}, {et~al\mbox{.}}} \bibinfo{year}{1996}\natexlab{}.
\newblock \showarticletitle{Columbia object image library (coil-20)}.
\newblock  (\bibinfo{year}{1996}).
\newblock


\bibitem[Pearson(1901)]%
        {pearson1901liii}
\bibfield{author}{\bibinfo{person}{Karl Pearson}.} \bibinfo{year}{1901}\natexlab{}.
\newblock \showarticletitle{LIII. On lines and planes of closest fit to systems of points in space}.
\newblock \bibinfo{journal}{\emph{The London, Edinburgh, and Dublin philosophical magazine and journal of science}} \bibinfo{volume}{2}, \bibinfo{number}{11} (\bibinfo{year}{1901}), \bibinfo{pages}{559--572}.
\newblock


\bibitem[Posamentier and Lehmann(2012)]%
        {posamentier2012secrets}
\bibfield{author}{\bibinfo{person}{Alfred~S Posamentier} {and} \bibinfo{person}{Ingmar Lehmann}.} \bibinfo{year}{2012}\natexlab{}.
\newblock \bibinfo{booktitle}{\emph{The secrets of triangles: a mathematical journey}}.
\newblock \bibinfo{publisher}{Prometheus Books}.
\newblock


\bibitem[Raudys et~al\mbox{.}(1991)]%
        {raudys1991small}
\bibfield{author}{\bibinfo{person}{Sarunas~J Raudys}, \bibinfo{person}{Anil~K Jain}, {et~al\mbox{.}}} \bibinfo{year}{1991}\natexlab{}.
\newblock \showarticletitle{Small sample size effects in statistical pattern recognition: Recommendations for practitioners}.
\newblock \bibinfo{journal}{\emph{IEEE Transactions on pattern analysis and machine intelligence}} \bibinfo{volume}{13}, \bibinfo{number}{3} (\bibinfo{year}{1991}), \bibinfo{pages}{252--264}.
\newblock


\bibitem[Richards and King(2014)]%
        {richards2014big}
\bibfield{author}{\bibinfo{person}{Neil~M Richards} {and} \bibinfo{person}{Jonathan~H King}.} \bibinfo{year}{2014}\natexlab{}.
\newblock \showarticletitle{Big data ethics}.
\newblock \bibinfo{journal}{\emph{Wake Forest L. Rev.}}  \bibinfo{volume}{49} (\bibinfo{year}{2014}), \bibinfo{pages}{393}.
\newblock


\bibitem[Rudin et~al\mbox{.}(1976)]%
        {rudin1976principles}
\bibfield{author}{\bibinfo{person}{Walter Rudin} {et~al\mbox{.}}} \bibinfo{year}{1976}\natexlab{}.
\newblock \bibinfo{booktitle}{\emph{Principles of mathematical analysis}}. Vol.~\bibinfo{volume}{3}.
\newblock \bibinfo{publisher}{McGraw-hill New York}.
\newblock


\bibitem[Safo(2014)]%
        {safo2014design}
\bibfield{author}{\bibinfo{person}{Sandra~Esi Safo}.} \bibinfo{year}{2014}\natexlab{}.
\newblock \emph{\bibinfo{title}{Design and analysis issues in high dimension, low sample size problems}}.
\newblock \bibinfo{thesistype}{Ph.\,D. Dissertation}. \bibinfo{school}{University of Georgia}.
\newblock


\bibitem[Saha(2009)]%
        {saha2009small}
\bibfield{author}{\bibinfo{person}{Sankha~Pallab Saha}.} \bibinfo{year}{2009}\natexlab{}.
\newblock \showarticletitle{On small sample prediction of financial crisis}. In \bibinfo{booktitle}{\emph{2009 Seventh International Conference on Advances in Pattern Recognition}}. IEEE, \bibinfo{pages}{43--46}.
\newblock


\bibitem[Sen and Bricka(2013)]%
        {sen2013data}
\bibfield{author}{\bibinfo{person}{Sudeshna Sen} {and} \bibinfo{person}{Stacey Bricka}.} \bibinfo{year}{2013}\natexlab{}.
\newblock \bibinfo{title}{Data collection technologies--Past, present, and future}.
\newblock
\newblock


\bibitem[Sharma et~al\mbox{.}(2018)]%
        {sharma2018analysis}
\bibfield{author}{\bibinfo{person}{Neha Sharma}, \bibinfo{person}{Vibhor Jain}, {and} \bibinfo{person}{Anju Mishra}.} \bibinfo{year}{2018}\natexlab{}.
\newblock \showarticletitle{An analysis of convolutional neural networks for image classification}.
\newblock \bibinfo{journal}{\emph{Procedia computer science}}  \bibinfo{volume}{132} (\bibinfo{year}{2018}), \bibinfo{pages}{377--384}.
\newblock


\bibitem[Shen et~al\mbox{.}(2016)]%
        {shen2016statistics}
\bibfield{author}{\bibinfo{person}{Dan Shen}, \bibinfo{person}{Haipeng Shen}, \bibinfo{person}{Hongtu Zhu}, {and} \bibinfo{person}{JS Marron}.} \bibinfo{year}{2016}\natexlab{}.
\newblock \showarticletitle{The statistics and mathematics of high dimension low sample size asymptotics}.
\newblock \bibinfo{journal}{\emph{Statistica Sinica}} \bibinfo{volume}{26}, \bibinfo{number}{4} (\bibinfo{year}{2016}), \bibinfo{pages}{1747}.
\newblock


\bibitem[Shen et~al\mbox{.}(2022)]%
        {shen2022classification}
\bibfield{author}{\bibinfo{person}{Liran Shen}, \bibinfo{person}{Meng~Joo Er}, {and} \bibinfo{person}{Qingbo Yin}.} \bibinfo{year}{2022}\natexlab{}.
\newblock \showarticletitle{Classification for high-dimension low-sample size data}.
\newblock \bibinfo{journal}{\emph{Pattern Recognition}}  \bibinfo{volume}{130} (\bibinfo{year}{2022}), \bibinfo{pages}{108828}.
\newblock


\bibitem[Shen et~al\mbox{.}(2023)]%
        {shen2023data}
\bibfield{author}{\bibinfo{person}{Xun Shen}, \bibinfo{person}{Hampei Sasahara}, \bibinfo{person}{Jun-ichi Imura}, \bibinfo{person}{Makito Oku}, {and} \bibinfo{person}{Kazuyuki Aihara}.} \bibinfo{year}{2023}\natexlab{}.
\newblock \showarticletitle{Data-driven re-stabilization of gene regulatory network towards early medical treatment}.
\newblock \bibinfo{journal}{\emph{IFAC-PapersOnLine}} \bibinfo{volume}{56}, \bibinfo{number}{2} (\bibinfo{year}{2023}), \bibinfo{pages}{6200--6205}.
\newblock


\bibitem[Shen et~al\mbox{.}(2024)]%
        {shen2024ultra}
\bibfield{author}{\bibinfo{person}{Xun Shen}, \bibinfo{person}{Naruto Shimada}, \bibinfo{person}{Hampei Sasahara}, {and} \bibinfo{person}{Jun-ichi Imura}.} \bibinfo{year}{2024}\natexlab{}.
\newblock \showarticletitle{Ultra-early medical treatment-oriented system identification using High-Dimension Low-Sample-Size data}.
\newblock \bibinfo{journal}{\emph{IFAC Journal of Systems and Control}} (\bibinfo{year}{2024}), \bibinfo{pages}{100245}.
\newblock


\bibitem[Strang(2019)]%
        {strang2019introduction}
\bibfield{author}{\bibinfo{person}{Gilbert Strang}.} \bibinfo{year}{2019}\natexlab{}.
\newblock \bibinfo{booktitle}{\emph{Introduction to linear algebra}}.
\newblock \bibinfo{publisher}{OpenCourseWare}.
\newblock
\urldef\tempurl%
\url{https://ocw.mit.edu/resources/res-5503-introduction-to-linear-algebra-fall-2011/}
\showURL{%
\tempurl}


\bibitem[Sun et~al\mbox{.}(2022)]%
        {sun2022effectiveness}
\bibfield{author}{\bibinfo{person}{Bo Sun}, \bibinfo{person}{Yang Zhang}, \bibinfo{person}{Qiming Zhou}, {and} \bibinfo{person}{Xinchang Zhang}.} \bibinfo{year}{2022}\natexlab{}.
\newblock \showarticletitle{Effectiveness of semi-supervised learning and multi-source data in detailed urban landuse mapping with a few labeled samples}.
\newblock \bibinfo{journal}{\emph{Remote Sensing}} \bibinfo{volume}{14}, \bibinfo{number}{3} (\bibinfo{year}{2022}), \bibinfo{pages}{648}.
\newblock


\bibitem[Tufte(1983)]%
        {tufte1983visual}
\bibfield{author}{\bibinfo{person}{Edward~R Tufte}.} \bibinfo{year}{1983}\natexlab{}.
\newblock \bibinfo{booktitle}{\emph{The visual display of quantitative information}}.
\newblock \bibinfo{publisher}{Graphic Design Books}.
\newblock


\bibitem[Tufte(1990)]%
        {tufte1990envisioning}
\bibfield{author}{\bibinfo{person}{Edward~R Tufte}.} \bibinfo{year}{1990}\natexlab{}.
\newblock \bibinfo{booktitle}{\emph{Envisioning information}}.
\newblock \bibinfo{publisher}{Graphics Press LLC}.
\newblock


\bibitem[Tufte(2001)]%
        {tufte2001visual}
\bibfield{author}{\bibinfo{person}{Edward~R Tufte}.} \bibinfo{year}{2001}\natexlab{}.
\newblock \bibinfo{booktitle}{\emph{The visual display of quantitative information}}.
\newblock \bibinfo{publisher}{Graphics Press LLC}.
\newblock


\bibitem[Van~der Maaten and Hinton(2008)]%
        {van2008visualizing}
\bibfield{author}{\bibinfo{person}{Laurens Van~der Maaten} {and} \bibinfo{person}{Geoffrey Hinton}.} \bibinfo{year}{2008}\natexlab{}.
\newblock \showarticletitle{Visualizing data using t-SNE.}
\newblock \bibinfo{journal}{\emph{Journal of machine learning research}} \bibinfo{volume}{9}, \bibinfo{number}{11} (\bibinfo{year}{2008}).
\newblock


\bibitem[Wang et~al\mbox{.}(2022)]%
        {wang2022enhanced}
\bibfield{author}{\bibinfo{person}{Xun Wang}, \bibinfo{person}{Zhiyong Yu}, \bibinfo{person}{Lisheng Wang}, {and} \bibinfo{person}{Pan Zheng}.} \bibinfo{year}{2022}\natexlab{}.
\newblock \showarticletitle{An enhanced priori knowledge GAN for CT images generation of early lung nodules with small-size labelled samples}.
\newblock \bibinfo{journal}{\emph{Oxidative Medicine and Cellular Longevity}}  \bibinfo{volume}{2022} (\bibinfo{year}{2022}).
\newblock


\bibitem[Ware(2004)]%
        {ware2004information}
\bibfield{author}{\bibinfo{person}{Colin Ware}.} \bibinfo{year}{2004}\natexlab{}.
\newblock \bibinfo{booktitle}{\emph{Information visualization: perception for design}}.
\newblock \bibinfo{publisher}{Morgan Kaufmann}.
\newblock


\bibitem[Weisstein({[n.\,d.]})]%
        {weissteinIdentity}
\bibfield{author}{\bibinfo{person}{Eric~W. Weisstein}.} \bibinfo{year}{[n.\,d.]}\natexlab{}.
\newblock \bibinfo{booktitle}{\emph{Identity Element}}.
\newblock
\urldef\tempurl%
\url{https://mathworld.wolfram.com/IdentityElement.html}
\showURL{%
\tempurl}
\newblock
\shownote{Visited on 18/02/24}.


\bibitem[Wertheimer(1923)]%
        {wertheimer1923laws}
\bibfield{author}{\bibinfo{person}{Max Wertheimer}.} \bibinfo{year}{1923}\natexlab{}.
\newblock \showarticletitle{Laws of organization in perceptual forms}.
\newblock \bibinfo{journal}{\emph{Psychol. Forsch.}}  \bibinfo{volume}{4} (\bibinfo{year}{1923}), \bibinfo{pages}{293--350}.
\newblock


\bibitem[Xiao et~al\mbox{.}(2017)]%
        {xiao2017fashion}
\bibfield{author}{\bibinfo{person}{Han Xiao}, \bibinfo{person}{Kashif Rasul}, {and} \bibinfo{person}{Roland Vollgraf}.} \bibinfo{year}{2017}\natexlab{}.
\newblock \showarticletitle{Fashion-mnist: a novel image dataset for benchmarking machine learning algorithms}.
\newblock \bibinfo{journal}{\emph{arXiv preprint arXiv:1708.07747}} (\bibinfo{year}{2017}).
\newblock


\end{thebibliography}

\end{document}